%% file: main.tex
\documentclass[11pt]{article}
\usepackage{amsthm, amsmath,amsfonts,graphicx, amssymb, bm,geometry}
\usepackage{xspace}
\usepackage{float,subcaption}

\usepackage[utf8]{inputenc} 

\usepackage{enumitem}

\usepackage[T1]{fontenc}    

\usepackage[usenames,dvipsnames]{xcolor}

\usepackage[hidelinks]{hyperref}
\usepackage{cleveref}
\usepackage{thmtools}
\crefname{figure}{Fig.}{Figs.}

\definecolor{darkpastelgreen}{rgb}{0.01, 0.75, 0.24}

\hypersetup{colorlinks={true},linkcolor={blue},citecolor=darkpastelgreen}

\usepackage[numbers]{natbib}
\usepackage[ruled]{algorithm2e} 

\theoremstyle{plain}
\newtheorem{theorem}{Theorem}[section]
\newtheorem{corollary}[theorem]{Corollary}
\newtheorem{lemma}[theorem]{Lemma}
\newtheorem{proposition}[theorem]{Proposition}

\theoremstyle{definition}
\newtheorem{definition}[theorem]{Definition}
\newtheorem{remark}[theorem]{Remark}
\newtheorem{example}[theorem]{Example}
\newtheorem{problem}{Problem}

\newcommand{\R}[0]{{\mathbb{R} }}
\newcommand{\X}[0]{{\mathcal{X} }}

\DeclareMathOperator{\supp}{{\rm supp }}

\newcommand{\Jen}[2][]{\ensuremath{\ifthenelse{\equal{#1}{}}{\mathfrak{J}(#2)}{\mathfrak{J}_{#1}(#2)}}\xspace}
\newcommand{\obj}[2][]{\ensuremath{\ifthenelse{\equal{#1}{}}{{\rm Obj }(#2)}{{\rm Obj}_{#1}(#2)}}\xspace}

\DeclareMathOperator{\E}{\mathbb{E}}

\newcommand{\defn}[1]{{\textbf{\textit{#1}}}}

\title{Optimal Scoring Rule Design under Partial Knowledge}
\author{Yiling Chen\thanks{Harvard University \texttt{yiling@seas.harvard.edu}}\and Fang-Yi Yu\thanks{George Mason University \texttt{fangyiyu@gmu.edu}}}

\date{}

\begin{document}

\maketitle
\begin{abstract}
\input{sections/abstract}
\end{abstract}

\input{sections/body}

\bibliographystyle{plainnat}
\bibliography{reference}
\newpage
\appendix
\input{sections/appendix}
\end{document}

%% file: sections/abstract.tex
This paper studies the design of optimal proper scoring rules when a principal has partial knowledge of an agent's signal distribution. Recent work~\cite{hartline2020optimization} characterizes the proper scoring rules that maximize the increase of an agent's payoff when the agent chooses to access a costly signal to refine a posterior belief from his prior prediction, under the assumption that the agent's signal distribution is fully known to the principal.    
In our setting, the principal only knows about a set of distributions where the agent's signal distribution belongs. We formulate the scoring rule design problem as a max-min optimization that maximizes the worst-case increase in payoff across the set of distributions.

We propose an efficient algorithm to compute an optimal scoring rule when the set of distributions is finite, and devise a fully polynomial-time approximation scheme that accommodates various infinite sets of distributions.   We further remark that widely used scoring rules, such as the quadratic and log rules, as well as previously identified optimal scoring rules under full knowledge~\cite{hartline2020optimization}, can be far from optimal in our partial knowledge settings.

%% file: sections/body.tex

\section{Introduction}

Proper scoring rules are scoring functions that incentivize truthful information elicitation: an agent with a subjective belief about an uncertain event maximizes his expected score by making a prediction according to his belief.  If the agent can acquire a costly signal to refine his belief, which proper scoring rules maximally incentivize the agent's information acquisition? Recent work~\cite{hartline2020optimization} explores this question under the assumption that the principal has full knowledge about the agent's information structure. In contrast, we investigate this optimal scoring rule design question when the principal has only partial knowledge about the agent's information structure, only knowing the set of information structures that the agent's belongs to.  

Incentivizing information acquisition is crucial in many real-world applications. For example, in conference reviewing, a reviewer could spend little to no effort and form his {prior assessment} about a paper primarily based on the length of the paper, the amount of grammatical errors in the introduction, or the references cited. But a conference chair would aspire for a high-effort review where the reviewer carefully reads the paper and then forms his, more informed, {posterior assessment}. As another example, in crowdsourced prediction for the reproducibility of scientific studies~\cite{Gordon2020,liu2020replication}, an expert could form a prediction on the replication outcome of a study based on his general knowledge about the study's topic, the reputation of the publication venue, and the authors' affiliations. But a more valuable and accurate prediction requires the expert to examine the study's methodology and evaluation procedures carefully. In these applications, the principal ideally hopes to devise a proper scoring rule that maximizes the expected increase in score if the agent acquires the information, to maximally incentivize information acquisition. The principal however only has limited knowledge about the agent's information structure.   



More formally, an agent's information structure consists of two parts: the agent's prior belief about the random variable of interest (the prior) and the distribution of the agent's signal conditioned on every realization of the random variable (an experiment). We formulate the optimal scoring rule design problem as a max-min optimization that maximizes the agent's worst-case increase in score across the set of possible information structures. We explore the problem for four settings of principal's knowledge, ranging from the special full-knowledge case to varying degree of partial knowledge:  

%

\begin{enumerate}
    \item The principal knows both the prior and the experiment and hence the set of information structures is a singleton.  We reprove \cite{hartline2020optimization}'s results and show that the optimal scoring rules are $v$-shaped in \cref{thm:singleton}. In addition, \cref{thm:singleton} also provides a close-form expression for an optimal scoring rule. 
    \item The principal knows the prior but is uncertain about the experiment. The set of information structures shares the same prior, as introduced in \cref{ex:uniprior}. Interestingly, we show in \cref{thm:uniprior} that the same $v$-shaped scoring rule in \cref{thm:singleton} is also optimal.
    \item The principal knows the experiment but is uncertain about the agent's prior.  The set of information structures shares the same experiment, as introduced in \cref{ex:uniexperi}. We present a fully polynomial time approximation scheme (FPTAS) that outputs approximately optimal piece-wise linear scoring rules in \cref{thm:uniexperi}.
    \item The principal is uncertain about both the prior and the experiment. In \cref{thm:finite}, we propose an efficient algorithm to compute an optimal scoring rule when the set of distributions is finite. Additionally, for $\rho$-correlated information structures defined in \cref{ex:beta}, we develop an FPTAS to find approximately optimal scoring rules in \cref{thm:uniexperi}. $\rho$-correlated information structures have interesting connections to Beta-Bernoulli model~\cite{murphy2012machine} and noise operator in Boolean function analysis~\cite{o2014analysis}.
\end{enumerate}

We then run simulations to evaluate the performance of two frequently used proper scoring rules (quadratic scoring rule and log scoring rule), the $v$-shaped scoring rule (which is optimal in the first two settings) and the piece-wise linear scoring rules obtained by our FPTAS for the $\rho$-correlated information structures.  The simulations show that our algorithm's piecewise linear scoring rules perform well, and provide the most uniform incentive when the signal and the state have a large correlation.  However, when the correlation is small, log scoring rule outperforms our piecewise linear scoring rules and other scoring rules.  In particular, as the correlation decreases, we observe that the piecewise scoring rule approaches the log scoring rule which may suggest an interesting connection between $\rho$-correlated information structure and the log scoring rule.  The $v$-shaped scoring rule empirically performs worst in our partial knowledge setting.  
\paragraph{Organization and contribution of technical results}
Using Savage's representation of proper scoring rules~\cite{mccarthy1956measures,savage1971elicitation,gneiting2007strictly},  we convert variable space of our max-min optimization problem from scoring rules to convex functions where the increase in payoff becomes the Jensen's gap of the associated convex function as \cref{prob:opt}.  

In \cref{sec:prior}, we derive a geometric interpretation of the optimization problem.  \Cref{lem:gain} shows that the information gain amounts to how \emph{curved} the associated convex function is at the prior of the information structure.  This interpretation is critical as it leads to \cref{thm:uniprior} that shows $v$-shaped scoring rule is optimal for known prior setting (\cref{ex:uniprior}).  

In \cref{sec:finite,sec:infinite}, we delve into the scenario of unknown prior settings. We first use linear program to find an optimal piecewise linear scoring rule when the collection of information structures is finite in \cref{sec:finite}.  Then we venture into the domain of infinite information, and provide an FPTAS (\cref{thm:uniexperi}) for various infinite sets  (\cref{ex:uniexperi,ex:beta}).  Informally, our FPTAS runs the linear program on a finite subset of the infinite set of information structures, and provides approximation guarantees as long as the finite subset is an $\epsilon$-covering of the original set under earth mover's distance.   To this end, we relate the earth mover's distance to our optimization problem in \cref{lem:lipschitz}.  Additionally, we design a novel coupling that can bound the earth mover's distance of two posterior predictions by the total variation distance of their joint distributions on signal and state in \cref{prop:tv2emd}.  This coupling argument may be of independent interest for designing approximation algorithms for information aggregation and elicitation.  

\paragraph{Related work}
Our problem can be seen as purchasing prediction from strategic people. The work on this topic can be roughly divided into two categories according to whether agents can misreport their signal or prediction.  Below we focus on the relationship of our work to the most relevant technical scholarship.

In the first category, to ensure agents reporting their signals or predictions truthfully, there are two settings according to whether money is used for incentive alignment.  In the first setting, the analyst uses monetary payments to incentivize agents to reveal their data truthfully.  The challenge is to ensure truth-telling gets the highest payments.  Existing works verify agents' reports by using either an observable ground truth (proper scoring rules) or peers' reports (peer prediction). 
For the second setting, individuals’ utilities directly depend on an inference or learning outcome (e.g. they want a regression line to be as close to their own data point as possible) and hence they have incentives to manipulate their reported data to influence the outcome.~\cite{dekel2010incentive,meir2012algorithms,perote2004strategy,hardt2015strategic,chen2018strategyproof}

Our setting generalizes \cite{hartline2020optimization}'s.  Our max-min optimization formulation captures the principal's partial knowledge about agents' information structure, while theirs focuses more on known information structure.  We consider both ex-ante and ex-post setting which allow us to compare our result to the log scoring rule which is arguably one of the most important proper scoring rules but cannot be ex-post bounded.  
They show the optimal scoring rule is $v$-shaped when the information structure is known.  We generalize the result to known prior setting in \cref{thm:uniprior}.

\cite{neyman2020binary} also study optimal scoring rule design problem, but is more related to sequential method~\cite{ghosh2011sequential}.  Instead of imposing bounded payment conditions, their objective comprises both payment and accuracy.  They consider a special case of Beta-Bernoulli information structures where the prior is uninformative and design scoring rule for an agent to sequentially acquire samples from a Bernoulli distribution.

Finally, \cite{https://doi.org/10.48550/arxiv.2204.01773} study an information acquisition problem: When the information structure and the cost of information are known, the principal chooses a proper scoring rule (menu of contracts) to incentivize information acquisition as cheaply as possible subject to limited liability.  Informally, their formulation can be seen as the dual of our problem. Instead of maximizing information gain subjected to bounded payment conditions, they want to minimize payment with a lower bound on the information gain.  Their optimal scoring rules are also $v$-shaped and similar to our known information structure setting.

In the second category, agents cannot misreport their signal or prediction.  The problem of purchasing data from people has been investigated with different focuses, e.g. privacy concerns~\cite{ghosh2011selling,fleischer2012approximately,ghosh2014buying,nissim2014redrawing,cummings2015truthful,waggoner2015market}, effort and cost of data providers~\cite{aaron2012conducting,cai2014optimum,abernethy2015actively,chen2017optimal,zheng2017active,chen2018prior}, and reward
allocation~\cite{ghorbani2019data,agarwal2019market}.

More generally, our problem is also related to contract theory~\cite{grossman1992analysis}.  Recent work also studies optimal contracts where the principal does not fully know the agent's cost.~ \cite{https://doi.org/10.48550/arxiv.2010.06742,https://doi.org/10.48550/arxiv.2111.09179} Moreover, \cite{bechtel2020delegated} study delegation problem that tries to optimize the efforts of others.  However, the critical difference that sets ours apart from these works is the principal's preference.  Most of the work in contract theory aims to maximize the principal's utility, but ours treats the principal's preference as a budget constraint and optimize the efforts of the agents.  Our formulation may be more suitable for complicated problems, e.g., peer grading or conference review, that consists of multiple sub-problems, and the principal's utility can not be easily decomposed according to sub-problems.  For the modelling choice, \cite{fu2014optimal} studies auction design where the unknown value distribution is from a known collection, without assuming a prior over the collection.
\section{Model and Preliminaries}\label{sec:pre}
In \cref{sec:pre1}, we first introduce proper scoring rules and two boundedness notions.  Then we define information structures that formalize the relationship between the costly signal and the event of interest, and we specify the principal's partial knowledge and introduce three motivating examples.  Finally, we define the value of information and our scoring rule design problem. 
We further simplify our problem by connecting scoring to convex function in \cref{sec:ps}.
\subsection{Optimal Scoring Rule for Costly Information}\label{sec:pre1}
This paper studies the design of scoring rules for binary state\footnote{We discuss the $d$-dimensional setting in the appendix.} which maps an agent's reported prediction $x$ and the realized ground state $\omega\in \{0,1\}$ to a score for the agents, $PS(\omega, x)$.  As a principal designs a scoring rule for the agent, one desirable scoring rule should elicit {truthful} predictions.  
\begin{definition}
    A scoring rule is \emph{proper} if for all predictions $x, x'\in [0,1]$,
$$\E_{\omega\sim x}[PS(\omega, x)]\ge \E_{\omega\sim x}[PS(\omega, x')].$$

In other words, if an agent's prediction for $\omega = 1$ is $x$, he cannot gain a higher score by misreporting $x'$.  Here $\omega\sim x$ denote random variable $\omega = 1$ with probability $x$, and $0$ otherwise.  
\end{definition}

Another desirable property of a scoring rule is boundedness as the principal has a finite budget for the reward.  Here we consider the following two notions.
\begin{definition}\label{def:bounded}
    A scoring rule $PS$ is \emph{ex-post bounded} by $B>0$ if for all $\omega\in \{0,1\}$ and $x\in [0,1]$, $PS(\omega, x)\in [0,B]$.  Alternatively, $PS$ is \emph{ex-ante bounded} by $B$ if for all $x\in [0,1]$ $\E_{\omega\sim x}[PS(\omega, x)]\in [0,B]$.
\end{definition}
In addition to properness and boundedness, the principal often wants to design proper scoring rules that incentivize effort.  Specifically, when the agent can refine his prediction by exerting a binary effort for a costly signal, how can the principal design a scoring rule that maximizes the agent's perceived gain from exerting effort without complete knowledge of the prior and posterior distributions?

\paragraph{Prior, posterior, and information structures}  The agent can access a costly signal $S\in \mathcal{S}$ with a finite support that improves his prediction of the unknown binary state of the world $W\in \{0,1\}$.  Specifically, the agent has an \defn{information structure} $P$ which is a joint distribution on the state and signal.  An information structure $P$ consists of a \emph{prior} $\pi\in [0,1]$ of the state and an \emph{experiment} $\sigma: \{0,1\}\to \Delta(\mathcal{S})$ which is a conditional distribution of the signal given the state, so that for all state $\omega\in \Omega$ and signal $s\in \mathcal{S}$, 
$$\pi = P(W = 1)\textrm{ and }\sigma(s|\omega) = P(S = s\mid W = \omega).$$   We will refer $P$ by the pair of prior and experiment $(\pi, \sigma)$.  

Given $P$ with $(\pi, \sigma)$, if the agent ignores the signal, his truthful prediction of the state is the \emph{prior}, $P(W = 1) = \pi$.  If the agent accesses the signal and sees $S = s$, his truthful prediction becomes the \emph{posterior}
$P(W = 1\mid S = s)
     = \frac{\pi\sigma(s|1)}{\pi\sigma(s|1)+(1-\pi)\sigma(s|0)}.$
The posterior prediction $X(s) := P(W = 1\mid S = s)$ is a random variable over the randomness of signal.  We will omit $s$ and write the random variable of posterior prediction as $X$.  For instance, the expectation of posterior equals prior $\E X = \pi$ by Bayes' rule.

\paragraph{Partial knowledge of information structures}
In our partial knowledge setting, the principal only knows a collection of information structures $\mathcal{P}$ that the agent’s information structure falls into.  Below are three examples of collections of information structures that model the principal's partial knowledge.  We will use these as running examples throughout the paper.

First the principal may know the signal distribution given ground state (experiment) but does not know the agent's background (prior). Intuitively, this captures online crowdsourcing setting, e.g., image annotation,  where the background of the agent is unknown, but quality of signal can be controlled.  Note that the priors are bounded away from zero and one by $\delta>0$, because if the prior is zero or one, the posterior is identical to the prior.   
\begin{example}\label{ex:uniexperi}
Given an experiment $\sigma$ and $0< \delta\le 1/2$, a collection of information structures with homogeneous experiment is
$\mathcal{P}_\sigma = \{(\pi, \sigma): \pi\in [\delta,1-\delta] \}.$
\end{example}

To another extreme, the principal may know agent's prior but does not know the agent's experiment.  This can model peer grading in classroom where the student's prior is known but the quality of his work is uncertain.  
\begin{example}\label{ex:uniprior}
Given a prior $\pi\in (0,1)$ and a set of experiments $\sigma_i$ where $i\in \mathcal{I}$, a collection of information structures with homogeneous prior is 
$\{(\pi, \sigma_i): i\in \mathcal{I}\}$.   
\end{example}

The first example has the same experiment for all information structures, and the second has the same prior.  We provide an example where the prior and experiment both vary.  
\begin{example}\label{ex:beta}
Given $\rho\in [0,1]$, $0< \delta\le 1/2$ and a binary signal space $\mathcal{S} = \{0,1\}$, a $\rho$-correlated experiment produces signal that equals the ground state with probability $\rho$ and samples from prior independently otherwise.  
A collection of $\rho$-correlated information structures $\mathcal{P}_\rho$ is
$$\{(\pi, \sigma): \pi\in [\delta, 1-\delta],  \sigma(1|1) = \rho+(1-\rho)\pi, \sigma(1|0) = (1-\rho)\pi\}.$$
Note that the value $\rho$ controls the correlation between the signal and the ground state.  In particular, if $\rho = 0$, the signal is independent of the ground state, and if $\rho = 1$, the signal perfectly agrees with the ground state.  Similar to \cref{ex:uniexperi}, the priors are bounded away from zero and one.
\end{example} 
There are several interesting interpretations of $\rho$-correlated experiments.  First, the information structure is the posterior predictive distribution of Beta-Bernoulli model~\cite{murphy2012machine} given an additional sample, and $\rho$ captures the strength of prior.  Second, the distribution of the signal and the state is also known as $\rho$-correlated pair in Boolean function analysis~\cite{o2014analysis}.  We formalize this connection in the appendix.

\paragraph{Information gain and max-min optimization problem}
Now we formalize the agent’s perceived gain from exerting effort under a proper scoring rule $PS$. 
Let $PS(x):=\E_{\omega\sim x}PS(\omega, x)$ be the expected score of truthful reporting a prediction $x$.  Given $PS$ and $P$ with $(\pi, \sigma)$, the expected score of the agent truthfully reporting his initial prediction is $PS(\pi)$.  Alternatively, if the agent accesses the costly signal $S = s$, he gets score $PS(X(s))$.  Thus, his expected score of accessing the costly signal before seeing it is $\E PS(X)$ over the randomness of $X$.  Since $\E X = \pi$, the agent's gain from exerting effort is 
$$\E PS(X)-PS(\pi) = \E PS(X)-PS(\E X).$$
We call the above \defn{information gain} on $P$ under the proper scoring rule $PS$. 
Because the information gain is fully determined by the posterior distribution random variable, we can use $P$ and $X$ interchangeably.  Additionally, we let $\mathcal{X}$ be the collection of random variables of prediction induced from information structures $P$ in $\mathcal{P}$. 

To maximize the agent's gain from exerting effort for any possible information structure in $\mathcal{X}$, the principal finds a bounded proper scoring rule $PS: \{0,1\}\times [0,1]\to \R$ which maximizes the worst-case information gain 
    $$\max_{PS}\min_{X\in \mathcal{X}}\E PS(X)-PS(\E X).$$
We will simplify our optimization in \cref{sec:ps} as \cref{prob:opt}.

\subsection{Savage's Representation of Proper Scoring Rules}\label{sec:ps}
We characterize the space of bounded proper scoring rules. First, Savage's representation of proper scoring rules connects proper scoring rules and convex functions. 

\begin{theorem}[\cite{mccarthy1956measures,savage1971elicitation}]\label{thm:ps}
When the state space $\Omega = \{0,1\}$ is binary, for every proper scoring rule $PS$, there exists a convex function $H:[0,1]\to\R$ so that for all $x\in [0,1]$ and $\omega\in \{0,1\}$
$$
PS(\omega, x) = \begin{cases}
H(x)+ \partial H(x)\cdot (1-x)&\text{ when } \omega = 1,\\
H(x)-\partial H(x)\cdot x&\text{ when } \omega = 0\end{cases}$$
where $\partial H(x)$ is a subgradient of $H$ at $x$.

Conversely, for every convex function $H: [0,1]\to \R$, there exists a proper scoring rule such that the above condition hold.
\end{theorem}

With above characterization, we can specify proper scoring rules by their associated convex functions. For instance, a \emph{quadratic scoring rule} has a convex function $2(x^2+(1-x)^2)-1$, and a \emph{log scoring rule} has a convex function $x\log_2 x+(1-x)\log_2 (1-x)+1$.
We will study \defn{piecewise-linear scoring rules} whose associated convex function is  piecewise linear $\max\{f_1(x), \dots, f_m(x)\}$ with affine functions $f_i$, $i = 1,\dots, m$.  Finally, \defn{$v$-shaped scoring rules} in~\cite{hartline2020optimization} are special cases of piecewise linear scoring rules,  whose corresponding convex function is $v$-shaped with parameter $(a,b,c,x_0)$ so that $H_{(a,b,c,x_0)}(x) = \max\{a(x-x_0)+c, b(x-x_0)+c\}$ with the \emph{vertex} at $x_0$.  \cref{fig:func2} shows examples of quadratic and log scoring rules, \cref{fig:v_shape} is for $v$-shaped scoring rules, and \cref{fig:linear} is for piecewise scoring rules.


Now we reformulate our optimization problem in terms of convex functions with the following lemmas by simple applications of \cref{thm:ps}.  The first lemma converts \emph{the information gain as the gap of Jensen's inequality of the associated convex function}.  The second lemma shows the bounded conditions can be checked by the value and sub-gradient of the convex function.  The proofs are in the appendix.
\begin{lemma}\label{lem:gain}
    For any proper scoring rule $PS$ with convex function $H$, the expected score of truthfully reporting $x$ is 
    $PS(x) = H(x)$, and the information gain of information structure $X$ under the proper scoring rule $PS$ is
    $\E PS(X)-PS(\E X) = \E[H(X)]-H(\E X).$
\end{lemma}
We will define the information gain of information structure $X$ under the proper scoring rule with $H$ as
\begin{equation}\label{eq:obj}
    \Jen[X]{H} := \E[H(X)]-H(\E X).
\end{equation}
\begin{lemma}\label{lem:bounded}
    Given $B>0$, for any proper scoring rule $PS$ with convex function $H$, $PS$ is ex-post bounded by $B$ if and only if
    $H(x)+\partial H(x)(1-x)$ and $H(x)-\partial H(x)x$ are in $[0,B]$ for all $x\in [0,1]$.
        
    $PS$ is ex-ante bounded by $B$ if and only if 
$H(x)\in [0,B]$ for all $x\in [0,1]$.
\end{lemma}
Let $\mathcal{B}_{\textrm{expost}}$ denote the set of ex-post bounded convex functions and $\mathcal{B}_{\textrm{exante}}$ for ex-ante bounded convex function as \cref{lem:bounded}.  By \cref{lem:bounded}, $\mathcal{B}_{\textrm{expost}}\subseteq \mathcal{B}_{\textrm{exante}}$.  
We now derive the simplified program for our max min optimization problem over the space of convex functions. 
\begin{problem}\label{prob:opt}
    Given a set of information structures $\mathcal{X}$ and a set of convex functions $\mathcal{B}$, find a convex function $H: [0,1]\to \R$ which maximizes the worst-case information gain 
    \begin{equation*}
        \max_{H\in \mathcal{B}} \min_{X\in \mathcal{X}}  \Jen[X]{H}.
    \end{equation*}
We will focus on $\mathcal{B}$ being $\mathcal{B}_{\text{expost}}$ in ex-post bounded setting and $\mathcal{B}_{\text{exante}}$ in ex-ante setting.
\end{problem}

\section{Main Results}
In \cref{sec:prior}, we explore the setting of known information structure and show $v$-shaped scoring rules are optimal in \cref{thm:singleton}.  We further show that the same $v$-shaped scoring rule is also optimal in the setting of known prior but uncertain about the experiment in \cref{thm:uniprior}.  
We shift our focus to the unknown prior setting.  We consider finite collections of information structures in \cref{sec:finite}, and devise an efficient algorithm that solves for the optimal scoring rules in \cref{thm:finite}.  Finally, we design an FPTAS in \cref{thm:uniexperi} for infinite collections of information structures, \cref{ex:uniexperi,ex:beta}.  

We provide outlines of proofs and intuitions, while complete proofs are deferred to the appendix.

\subsection{Singleton and Homogeneous Prior Information Structures}\label{sec:prior}

As a warm-up, let's consider the principal exactly knows the agent's information structure so that $\mathcal{X} = \{X\}$ is a singleton.  We show that the optimal $H$ can be $v$-shaped (defined in \cref{sec:ps}) with the vertex at the prior in both ex-ante and ex-post bounded settings.


\begin{figure}
    \centering
\includegraphics[width=0.7\columnwidth]{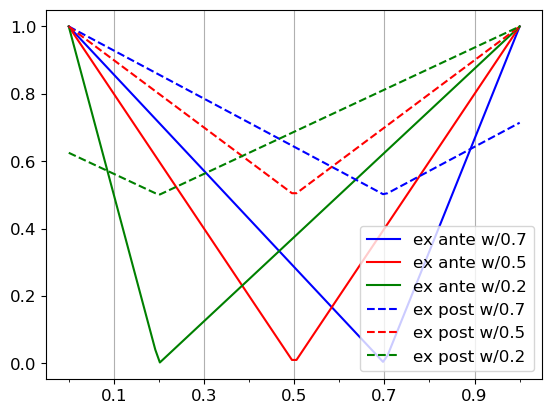} 
    \caption{The optimal scoring rules for the homogeneous prior setting. With \cref{thm:uniprior}, the dashed lines are the optimal $v$-shaped convex functions for three different priors ($0.7, 0.5$, and $0.2$) in the ex-post setting, and the solid ones are in the ex-ante setting.}
    \label{fig:v_shape}
\end{figure}
\begin{proposition}[singleton]\label{thm:singleton}
If $\mathcal{X} = \{X\}$ is singleton in ex-post or ex-ante bounded settings, there exists an optimal scoring rule associated with a $v$-shaped convex function for \cref{prob:opt}.  
\begin{enumerate}
    \item A $v$-shaped convex function with $x_0 = \pi$, $a = \frac{-B}{\pi}, b = \frac{B}{1-\pi}$ and $c = 0$ is optimal in ex-ante bounded setting with $B$.
    \item A $v$-shaped convex function with $x_0 = \pi$, $a = \frac{-B}{2\max \{\pi, 1-\pi\}}$, $b = \frac{B}{2\max \{\pi, 1-\pi\}}$, and $c = \frac{B}{2}$ is optimal in ex-post bounded setting with $B$.
\end{enumerate}
\end{proposition}
While the above result is already proved in \cite{hartline2020optimization}, we provide an explicit closed-form expression of the optimal solution and present an alternative proof in the appendix.  We generalize the proof and show that $v$-shaped scoring rule is optimal for any collection of information structure when they share the same prior (defined in \cref{ex:uniprior}).

\begin{theorem}[homogeneous prior]\label{thm:uniprior}
Given any collection of information structures with homogeneous prior $\pi$ (\cref{ex:uniprior}), there exists an optimal scoring rule associated with a $v$-shaped convex function with the vertex at $\pi$ in both ex-ante and ex-post bounded settings respectively.
\end{theorem}
The proof follows from the fact that the optimal $v$-shaped scoring rules only depend on the prior.  When multiple information structures have the same prior, the same $v$-shaped scoring rule is still optimal.  
\begin{figure}
    \centering
\includegraphics[width=0.7\columnwidth]{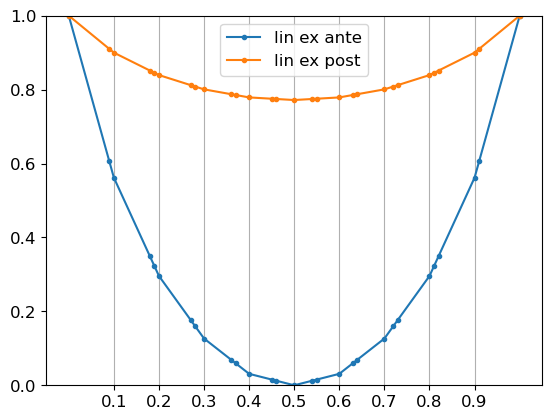}
    \caption{The optimal piecewise linear functions for a finite collection of information structures $\mathcal{P}_{\rho, N} = \{(\pi, \sigma)\in \mathcal{P}_\rho : N\pi\in \mathbb{N}\}$ with $\rho = 0.1$ and $N = 10$ (\cref{sec:sim}). In both ex-post and ex-ante setting, the vertex of the function is at the prior of the collection of information structures which is aligned with our intuition in \cref{sec:infinite} which suggests maximizing curvature at the prior.}
    \label{fig:linear}
\end{figure}

These results suggest that the principal should choose $H$ that is ``curved'' at the prior in order to incentivize the agent to derive the signal and move away from the prior as \cref{fig:v_shape}.  


\subsection{Finite Collections of Information Structures}\label{sec:finite}
If the collection of information structures is \emph{finite} so that $|\X|<\infty$, we give an efficient algorithm to compute an optimal piecewise linear scoring rule.   
Finite collections of information structures are natural when there are finite types of agents, and is useful to approximate some infinite collections of information structures as shown in the next section.



\begin{theorem}\label{thm:finite}
If $\X$ is finite, for both bounded settings, there exists an optimal proper scoring rule that is piecewise linear and can be derived by solving a linear program in time polynomial in $|\X|\cdot|\mathcal{S}|$.
\end{theorem}

The main idea is that when $\X$ is finite $\Jen[X]{H}$ in \cref{eq:obj} only depends on the evaluations of $H$ on the support of $\X$ $\overline{\supp}(\X) := \cup_{X\in \X} (\supp(X)\cup \{\E X\})$.  Thus, instead of searching all possible bounded convex functions, we can reduce the dimension of \cref{prob:opt} and use a linear program whose variables contain the evaluations of $H$ in $\overline{\supp}(\X)$ and linear constraints ensure that the evaluations can be extended to a convex piecewise linear function.  \Cref{fig:linear} presents an example of piece-wise linear convex function outputted by our linear program.  This observation allows us to solve the problem in weakly polynomial time in $|\overline{\supp}(\X)| \le |\X|\cdot|\mathcal{S}|$.  We present the formal proof in the appendix.

\subsection{Infinite Collections of Information Structures}\label{sec:infinite}

For a non-finite collection of information structures, solving \cref{prob:opt} exactly is generally infeasible, because even evaluating the objective value $\min_{X \in \mathcal{X}} \Jen[X]{H}$ may depend on evaluating all $X \in \mathcal{X}$, which could require unbounded time.  In particular, this limitation also applies to \cref{ex:uniexperi,ex:beta}.  Below, we extend our algorithms for finite collections of information structures to more general collections. We demonstrate the idea by approximately solving the optimal scoring rule for \cref{ex:uniexperi,ex:beta} in \cref{thm:uniexperi}.  \Cref{rem:general} discusses the potential and limitation of our method for abstract collections of information structures. 

\begin{theorem}\label{thm:uniexperi}  Given any $\epsilon>0, \delta>0$ and $B>0$, in ex-post bounded setting $\mathcal{B} = \mathcal{B}_{expost}$ (or ex-ante setting $\mathcal{B} = \mathcal{B}_{exante}$), there exists an efficient algorithm that outputs an $\epsilon$-optimal scoring rule $H$ on information structures with a homogeneous experiment (in \cref{ex:uniexperi}) so that, 
$$\min_{X\in \X}  \Jen[X]{H}\ge \max_{H'\in \mathcal{B}}\min_{X\in \X}\Jen[X]{H'}-\epsilon,$$
with running time polynomial in $B$ and $1/\epsilon$ (or $B$, $1/\epsilon$ and $1/\delta$ respectively). 
The same results hold for $\rho$-correlated information structures in \cref{ex:beta}.
\end{theorem}

Our FPTAS computes a finite collection of information structures, and runs our linear program in \cref{thm:finite}.  
Recall that $\mathcal{P}_\sigma$ is the collection of information structures with homogeneous experiment (in \cref{ex:uniexperi}), our algorithm picks a finite set information structures 
\begin{align*}
    \mathcal{P}_{\sigma, N} = \{(\pi, \sigma)\in \mathcal{P}_\sigma :N\pi\in \mathbb{N}\}
\end{align*} with $N = \lfloor B/(2\epsilon)\rfloor$ and outputs an optimal piecewise linear scoring rule $H_{\sigma, N}$ for $\mathcal{P}_{\sigma, N}$ using \cref{thm:finite}.  Similarly, for $\rho$-correlated information structures, let
$\mathcal{P}_{\rho, N} = \{(\pi, \sigma)\in \mathcal{P}_\rho : N\pi\in \mathbb{N}\}$.  Our algorithm outputs an optimal piecewise linear scoring rule $H_{\rho, N}$ for $\mathcal{P}_\rho$.  

The main challenge is to show the approximation guarantees.  We observe that given a pair of information structures, the difference between information gains should be small if their posterior distribution is close. Formally,  given two posteriors $X$ and $X'$ induced by two information structures $P$ and $P'$ respectively, the \defn{earth mover's distance} (EMD)~\cite{Chatterjee2008} between posteriors is
$d_{EM}(X,X') = \sup_{f:\|f\|_{Lip}\le 1}\E[f(X)]-\E[f(X')]$
where $\|f\|_{Lip}$ is the Lipschitz constant of $f$.
The following lemma relates the earth mover's distance to our optimization problem.  
\begin{lemma}\label{lem:lipschitz}
    Consider two posteriors $X$ and $X'$, and $B>0$. For any $H\in \mathcal{B}_{expost}$,
    $|\Jen[X]{H}-\Jen[X']{H}|\le 2B\cdot d_{EM}(X,X').$
    Similarly, if $\supp(\X)$ and $\supp(\X')$ are contained in $[1/L, 1-1/L]$ for some $L>0$, for any $H\in \mathcal{B}_{exante}$, 
    $|\Jen[X]{H}-\Jen[X']{H}|\le 2LB\cdot d_{EM}(X,X').$
\end{lemma}

To use the above results, we need to show $\mathcal{P}_{\sigma, N}$ is a good \emph{covering} for $\mathcal{P}_\sigma$ so that for all information structures in $\mathcal{P}_\sigma$ there exists one in $\mathcal{P}_{\sigma, N}$ that has a small earth mover's distance on the posteriors.  Here, we design a novel coupling that can bound the earth mover’s distance of two posterior predictions by the total variation distance on signal and state space.  
Given two information structures $P$ and $P'$ on $\Omega\times \mathcal{S}$, the \defn{total variation distance} (TVD)~\cite{Chatterjee2008} between these information structures is $d_{TV}(P,P') = \frac{1}{2}\sum_{w,s}|P(w,s)-P'(w,s)|.$


\begin{lemma}\label{prop:tv2emd}
    Given two information structures $P$ and $P'$ with induced posterior predictions $X$ and $X'$, 
    $d_{EM}(X, X')\le d_{TV}(P, P').$
\end{lemma}
\Cref{prop:tv2emd} upper bound the EMD of posterior predictions by the TVD of information structures.  The proof use the maximal coupling of $P$ and $P'$to couple $X$ and $X'$ using the maximal coupling of $P$ and $P'$.  Then we show the expected difference of $X$ and $X'$ in our coupling can be converted to the total variation distance.  Note that the EMD between two random variables is always smaller than the TVD between them, but \cref{prop:tv2emd} uses the TVD of the information structures, which can be much smaller than the TVD of posterior.\footnote{For instance, the TVD of posterior distribution between two $\rho$-correlated information structures with prior $\pi$ and $\pi'$ respectively is $1$, but the TVD of information structures is less than $2|\pi-\pi'|$ as shown in \cref{lem:covering}.}

\begin{lemma}\label{lem:covering}
    Given $\delta>0$, $N$, and $\sigma$, for all $P\in \mathcal{P}_\sigma$, there exists $P'\in \mathcal{P}_{\sigma, N}$ so that $d_{TV}(P, P')\le \frac{1}{N}$.
    Similarly, given $\rho>0$, for all $P\in \mathcal{P}_\rho$, there exists $P'\in \mathcal{P}_{\rho, N}$ so that $d_{TV}(P, P')\le \frac{2-\rho}{N}$.  
    $\overline{\supp}(\mathcal{P}_{\sigma, N}) = O(N|\mathcal{S}|)$, and $\overline{\supp}(\mathcal{P}_{\rho, N}) = O(N)$.
\end{lemma}

\Cref{lem:covering} shows $\mathcal{P}_{\sigma, N}$ and $\mathcal{P}_{\rho, N}$ are good coverings for $\mathcal{P}_\sigma$ and $\mathcal{P}_{\rho}$ in TVD respectively, by direct computation.

\begin{proof}[Proof of \cref{thm:uniexperi}]
    Given $\delta, \epsilon>0$, and  $\mathcal{P}_{\sigma} = \{(\pi, \sigma): \pi\in [\delta, 1-\delta]\}$ in \cref{ex:uniexperi}, we run our algorithm in \cref{thm:finite} on a finite collection of information structures $\mathcal{P}_{\sigma, N} = \{(\pi, \sigma):\pi\in [\delta, 1-\delta], N\pi\in \mathbb{N}\}$ with an integer $N = \lceil 2B/\epsilon\rceil $.  Because $|\mathcal{P}_{\sigma, N}| = O(B/\epsilon)$, the running time is polynomial in $B/\epsilon$ and $|\mathcal{S}|$.  For the approximation guarantee, given any ex-post bounded $H'\in \mathcal{H}_{expost}$ and $X\in \mathcal{P}_{\sigma}$, there exists $X'\in \mathcal{P}_{\sigma, N}$ so that 
    $$d_{EM}(X, X')\le \frac{\epsilon}{2B}$$ by \cref{lem:covering,prop:tv2emd}.  Moreover, by \cref{lem:lipschitz} and $N\ge 2B/\epsilon$
    $$|\Jen[X]{H'}-\Jen[X']{H'}|\le 2B\cdot d_{EM}(X,X')\le \epsilon.$$
    Therefore, the optimal $H$ for $\mathcal{P}_{\sigma, N}$ satisfies that
    $$\min_{X\in \X} \Jen[X]{H'}\le \min_{X\in \X'} \Jen[X]{H'}+\epsilon\le \min_{X\in \X'} \Jen[X]{H}+\epsilon$$
    which completes the proof for the ex-post bounded case.  The proof for ex-ante bounded case is similar to the above.
\end{proof}

\begin{remark}\label{rem:general}
Our FPTAS runs the linear program on a finite subset of an infinite set of information structures and provides approximation guarantees as long as the finite subset is an $\epsilon$-covering of the original set under earth mover’s distance by \cref{lem:lipschitz}.  The approach easily extends to any set with efficiently computable small coverings.\footnote{By \cite[Theorem 2.2.11]{panaretos2020invitation}, there always exists an $\epsilon$-covering with size $O(1/\epsilon)$.}  Furthermore, the optimization \cref{prob:opt} can be seen as zero-sum games between a designer choosing a scoring rule and the nature picking an information structure.~\cite{guo2024algorithmic}  Hence, if a best response oracle on $\Theta$ or an $\epsilon$-covering is available, various no-regret and best response algorithms can be employed to solve it.

However, the main challenge of solving \cref{prob:opt} lies in understanding the set $\Theta$.  Unlike standard optimization problems, $\Theta$ is often non-convex, including \cref{ex:uniexperi,ex:beta}.  Consequently, a general efficient algorithm assuming the existence of efficiently computable small coverings or best response oracle may appear vacuous.  Instead, we believe our method of finding and proving the small covering property could be a valuable tool for future research in analyzing their own $\Theta$.

\end{remark}

\section{Simulations}\label{sec:sim}
We compare the performance of common scoring rules (\cref{sec:ps}) within the context of \cref{prob:opt}.  We will focus on $\rho$-correlated information structures (\cref{ex:beta}) in the ex-ante setting.  This collection of information structures offers a platform to evaluate these scoring rules under the uncertainty of prior (and the experiment to less extent).  The other combinations (homogeneous experiment, or ex-post setting) yield comparable outcomes and are deferred to the appendix.  

We consider the correlation of signal and the state is either small $\rho = 0.025$ or large $\rho = 0.25$, and the prior $\pi$ is in $[0.01, 0.99]$.  For proper scoring rules, we use 1) a quadratic scoring rule, 2) a log scoring rule, 3) a $v$-shaped scoring rule with $(a, b, c, x_0) = (-2, 2, 0, 0.5)$ which is optimal when prior is $0.5$ as \cref{thm:uniprior}, and 4) the piecewise linear scoring rule $H_{\rho, N}$ for $\mathcal{P}_{\rho, N}$ derived by our algorithm in \cref{thm:uniexperi} with $N = 50$.  \Cref{fig:func2} shows the associated convex functions for these five proper scoring rules.  We include these scoring rules to investigate three questions: How good are classic scoring rules in our setting (the quadratic and log scoring rules)?  How does misspecified prior harm the optimal scoring rule in the known prior setting ($v$-shaped scoring rule)?  How does our FPTAS generalizes (piecewise linear scoring rules)?

\begin{figure}
\centering
    \begin{minipage}[t]{.32\textwidth}
        \includegraphics[width=\textwidth]{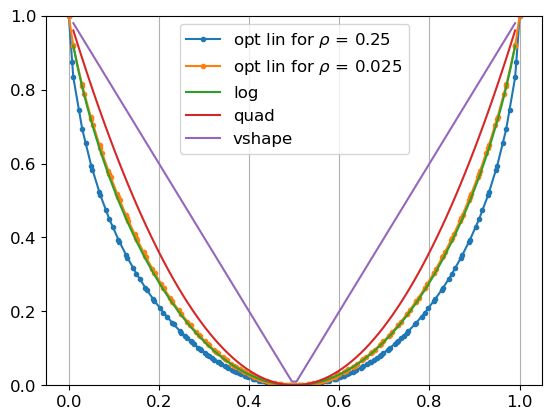}
        \subcaption{Associated convex functions}\label{fig:func2}
    \end{minipage}
    \hfill
    \begin{minipage}[t]{.32\textwidth}
        \includegraphics[width=\textwidth]{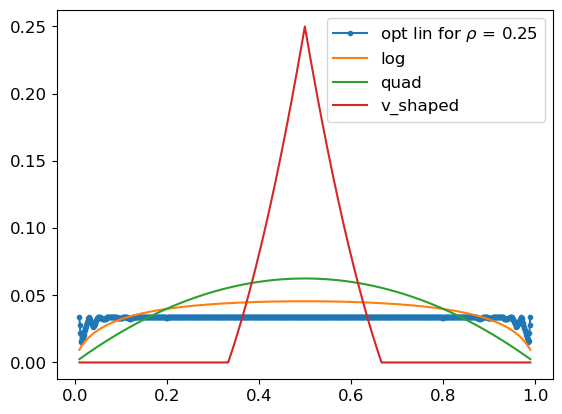}
        \subcaption{Information gain with $\rho = 0.25$.}\label{fig:large}
    \end{minipage}
    \hfill
    \begin{minipage}[t]{.32\textwidth}
        \includegraphics[width=\textwidth]{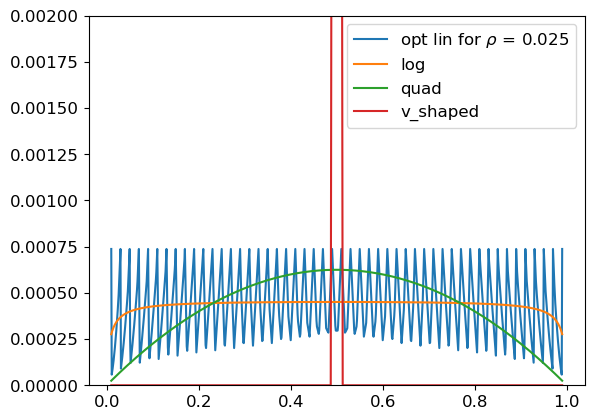}
        \subcaption{Information gain with $\rho = 0.025$.}\label{fig:small}
    \end{minipage}  
    \label{fig:1-2-3}
    \caption{\Cref{fig:func2} shows the five associated convex functions in in \cref{sec:sim}.  By \cref{lem:bounded} all scoring rules are ex-ante bounded by $B = 1$.  \Cref{fig:large,fig:small} present the information gain on each $\rho$-correlated information structures.  $x$ axis is the prior and $y$ axis is the information gain.}
\end{figure}


To begin with, we evaluate all our scoring rules with the information gain on each information structure in $\mathcal{P}_{\rho, N'}$ with $N' = 1000>N$ that serves as a proxy of the infinite collection $\mathcal{P}_\rho$ and is a superset of $\mathcal{P}_{\rho, N}$.  \Cref{fig:large,fig:small} present the outcomes where the $y$-axis represents the information gain, while the $x$-axis denotes the position of the prior.  Note that our piecewise linear scoring rules are optimal for $\mathcal{P}_{\rho, N}$, but only $\epsilon$-optimal for  $\mathcal{P}_{\rho, N'}$ with $\epsilon \le \frac{2-\rho}{N}\le 0.04$ by \cref{thm:uniexperi}.  The difference of information gains between $\mathcal{P}_{\rho, N}$ and $\mathcal{P}_{\rho, N'}$ measure how well our method generalizes.  

\Cref{fig:large} shows the information gain on $\rho$-correlated information structures with large correlation $\rho = 0.25$.  The piecewise linear scoring rule for $\mathcal{P}_{\rho, N}$ has the worst-case information gain  $0.0341$ on the original set $\mathcal{P}_{\rho, N}$, and generalizes well on the superset $\mathcal{P}_{\rho, N'}$ with  the worst-case information gain $0.0149$.  Comparatively, log scoring rule gets $0.0094$, the quadratic scoring rule gets $0.0024$, and the $v$-shaped scoring rule gets $0$ worst-case information gain on $\mathcal{P}_{\rho, N'}$.  The $v$-shaped scoring rule performs varies significantly: It achieves the highest information gain on information structures with prior centered around $0.5$ aligning with the implications of  \cref{thm:uniprior}, but also gets zero information gain when prior is above $0.7$ or below $0.3$.  This is because when the prior is far away from the vertex ($0.5$), information structure can have a support contained in a flat area $(0, 0.5)$ or $(0.5, 1)$ and leading to zero information gain because the Jensen's inequality is tight on affine functions.

\Cref{fig:small} shows the information gains when the correlation is small $\rho = 0.025$.  The piecewise linear scoring rule has the worst-case information gain $5.77\cdot 10^{-5}$ compared to $2.76\cdot 10^{-4}$ under the log scoring rule, $2.48\cdot 10^{-5}$ on quadratic scoring rule, and $0$ under the $v$-shaped scoring rule.  First, the information gains are much small in the small correlation setting, because the posterior is barely move away from prior.\footnote{Indeed, we can show that a correlated signal is Blackwell dominated by a correlated signal with larger correlation, and thus has a smaller information gain under any proper scoring rules.}  Second, information gains under the piecewise linear scoring rule on $\mathcal{P}_{\rho, N'}$ has several periodic peaks each peak.  Note that the piecewise linear can be seen as several small $v$-shaped scoring rules.  If an information structure's prior is at a vertex, the information gain is large.  However, if the support of the  information structure is contained in a flat area, the information gain is near zero.  Noteworthy, the number of vertices of our optimal scoring rule equal $N+1 = 51$ by \cref{thm:finite} which is also the number of peaks in the information gains.  On the other hand, the log and quadratic scoring rule are strictly convex which do not have any sharp transition between flat area and vertex, so the information gains change smoothly as the priors change. 
\begin{figure}
    \centering
\includegraphics[width=0.7\columnwidth]{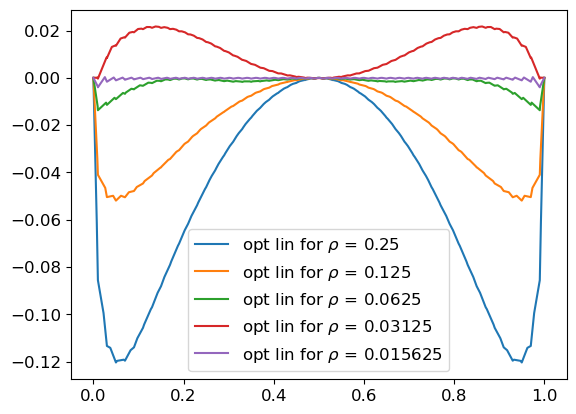}
    \caption{The difference between log scoring rule and the piecewise optimal scoring rules for $\rho$-correlated information structures $\mathcal{P}_{\rho, N}$ as $\rho$ ranging from $1/4$ to $1/64$.}\label{fig:diff}
\end{figure}
Finally, observe that the piecewise linear scoring rules are surprisingly close to the log scoring rule in \cref{fig:func2}.  In \cref{fig:diff}, we conduct an additional simulations and observe that difference between the piecewise scoring rule for $\mathcal{P}_{\rho, N}$ and the log scoring rule uniformly decreases as $\rho$ decreases.  This results suggest an interesting potential connection between $\rho$-correlated information structure and the log scoring rule.



\section{Conclusion}
We propose the problem of optimal proper scoring rules design when the principal has partial knowledge of an agent's signal distribution.  As shown in our simulations, this approach may serve as a benchmark for various scoring rules.  We devise efficient algorithms for four principal knowledge settings, and design a novel coupling to bound the earth mover's distance of posterior by the total variation distance of the signal and state space.


%% file: sections/appendix.tex
\section{Additional Details for Section~\ref{sec:pre}}
\subsection{Basic Properties of Proper Scoring Rules}

\begin{proof}[Proof for \cref{lem:gain}]
By \cref{thm:ps}, the expected scoring of truthfully reporting $x$ is 
\begin{align*}
    &\E_{\omega\sim x}PS(\omega, x)\\
    =& P[\omega = 1] \left(H(x)+\partial H(x)\cdot (1-x)\right)\\
    &+P[\omega = 0]\left(H(x)-\partial H(x)\cdot x\right)\\
    =& x \left(H(x)+\partial H(x)\cdot (1-x)\right)\\
    &+(1-x)\left(H(x)-\partial H(x)\cdot x\right)\\
    =& H(x)
\end{align*}
Thus, the information gain is $\E[H(X)]-H(\E X)$.
\end{proof}

\begin{proof}[Proof of \cref{lem:bounded}]
The proof straightforwardly follows from \cref{lem:gain} and \cref{thm:ps}.
\end{proof}

\begin{lemma}\label{lem:affine}
    Given a proper scoring rule with $H:[0,1]\to \mathbb{R}$ and an affine function $g(x) = \alpha x+\beta$, for any information structure $X$, the gain on $H+g$ where $(H+g)(x) = H(x)+g(x)$ for all $x$ is identical to the gain on $H$,
    $$\Jen[X]{H+g} = \Jen[X]{H}.$$
    If $\alpha>0$, the gain on $g\circ H$ where $g\circ H(x) = g(H(x))$ for all $x$ is the gain on $H$ scaled by $\alpha$, 
    $$\Jen[X]{g\circ H} = \alpha\Jen[X]{H}.$$
\end{lemma}
\begin{proof}
    Because \cref{eq:obj} is linear on the convex function and $\Jen[X]{g} = 0$ by Jensen's inequality on affine functions, $\Jen[X]{H+g} = \Jen[X]{H}+\Jen[X]{g} = \Jen[X]{H}$.  

    $\Jen[X]{g\circ H} = \Jen[X]{\alpha H+\beta} = \alpha \Jen[X]{H}$ by the linearity of information gain. 
\end{proof}
\subsection[Connection to Beta-Bernoulli Model]{Beta-Bernoulli and $\rho$-correlated information structures}
We formalize the connection between $\rho$-correlated information structures and Beta-Bernoulli distributions.


Suppose we want to collect predictions on an outcome of a coin $W$.  Given $p\in [0,1]$, the outcome $W$ follows the Bernoulli distribution $Bern(p)$ with $\Pr[W = 1] = p$, whereas the true value of $p$ is unknown.  If the agent privately observes $N$ i.i.d samples from the coin, how can we incentivize the agent to make one additional observation?


Given $\pi\in [0,1]$, $\mathcal{S} = \{0,1\}$, and $m\in \mathbb{N}_{>0}$, let $Beta(\pi, m)$ be a Beta distribution where the probability density at $p\in [0,1]$ is proportional to $p^{\pi m-1}(1-p)^{(1-\pi)m-1}$.  Here $\pi$ is the mean and $m$ is the effective sample size.  We define a \emph{Beta-Bernoulli information structure} with $\pi, m$ as follow: $p$ is sampled from $Beta(\pi, m)$.  $W$ and $S$ are sampled independently and identically from $Bern(p)$.

The agent starts with an uninformative prior,\footnote{the parameter of the coin $p$ is a uniform distribution on $[0,1]$} and  observes $n_1$ heads and $n-n_1$ tails.  By Bayes' rule, the agent's prior on $W = 1$ is $\pi = \frac{n_1+1}{n+2}$, and the posterior predictive distribution is $X = P(W = 1\mid S)$ that is a Beta-Bernoulli information structure with $(\pi, n+2)$, 
    $$X(s) =\begin{cases}
    \frac{(n+2)\pi+1}{n+3} &s = 1\text{ with probability }\pi\\
    \frac{(n+2)\pi}{n+3} &s = 0\text{ with probability } 1-\pi
    \end{cases}.$$
Alternatively, the conditional distribution of signal (experiment) satisfies $\sigma(1|1) = (1-\frac{1}{n+3})\pi+\frac{1}{n+3}$ and $\sigma(1|0) = (1-\frac{1}{n+3})\pi$.  Therefore, the Beta-Bernoulli information structure with $\pi, n+2$ is $(\frac{1}{n+3})$-correlated information structure with prior $\pi$.

\section{Proof and Details for Section~\ref{sec:prior}}

The following lemma shows that replacing $H$ by a $v$-shaped $\tilde{H}$ with the vertex at the prior has a better or equal information gain.  
\begin{lemma}\label{lem:vshape}
    Given an information structure $X$ with prior $\pi\in (0,1)$ and a scoring rule with $H$, there exists a $v$-shaped scoring rule with $\tilde{H} = H_{a, b, c, x_0}$ where $x_0 = \pi$, $a = \frac{H(\pi)-H(0)}{\pi}$, $b = \frac{H(1)-H(\pi)}{1-\pi}$, and $c = H(\pi)$, so that
    $$\Jen[X]{\tilde{H}}\ge \Jen[X]{H}.$$
    Additionally, if $H$ is in $\mathcal{B}_{exante}$ or $\mathcal{B}_{expost}$, $\tilde{H}$ is in $\mathcal{B}_{exante}$ or $\mathcal{B}_{expost}$ respectively.  
\end{lemma}
\begin{proof}[Proof of \cref{lem:vshape}]
    We first show the information gain of the $v$-shaped scoring rule is no less than the original scoring rule.  Then we show the $v$-shaped scoring is also bounded. 

    By definition and convexity, $\tilde{H}(\pi) = H(\pi)$, $\tilde{H}(1) = H(1)$, and $\tilde{H}(0) = H(0)$.  Because $H$ is convex, for all $x$, $H(x)\le \tilde{H}(x)$ and $H(x)-H(\pi)\le \tilde{H}(x)-\tilde{H}(\pi)$.  Therefore, 
    \begin{align*}
        \Jen[X]{H} =& \E[H(X)]-H(\E X)\\
        =& \E[H(X)-H(\pi)]\\
        \le& \E[\tilde{H}(X)-\tilde{H}(\pi)]\\
        =& \Jen[X]{\tilde{H}}
    \end{align*}
    which completes the first part.  

    \textbf{ex-ante bounded}
    On the other hand, if $H$ is ex-ante bounded with $B$, $H(x)\in [0,B]$ for all $x\in [0,1]$.  Then $\tilde{H}(x)\ge H(x)\ge 0$ for all $x$.  For the upper bound, $\tilde{H}(x)\le \max \tilde{H}(0), \tilde{H}(1)\le \max H(0), H(1)\le B$ by Jensen's inequality.  

    \textbf{ex-post bounded}
    If $H$ is ex-post bounded by $B$, $H(0), H(1) \le B$, $H(0)+\partial H(0), H(1)-\partial H(1)\ge 0$, and $H(\pi)\in [0,B]$ by \cref{lem:bounded}.  By the definition of $v$-shaped scoring rules and \cref{thm:ps}, there are five possible values of scores
    \begin{align*}
        \widetilde{PS}(x, \omega) =&  \begin{cases}-ax_0+c\text{ if }x\le x_0, \omega = 0,\\
            -bx_0+c\text{ if }x> x_0, \omega = 0,\\
            a(1-x_0)+c\text{ if }x\le x_0, \omega = 1,\\
            b(1-x_0)+c\text{ if }x> x_0, \omega = 1\end{cases}\\
            =& \begin{cases}
            H(0)\text{ if }x\le x_0, \omega = 0,\\
            H(1)-\frac{H(1)-H(\pi)}{1-\pi}\text{ if }x> x_0, \omega = 0,\\
            H(0)+\frac{H(\pi)-H(0)}{\pi}\text{ if }x\le x_0, \omega = 1,\\
            H(1)\text{ if }x> x_0, \omega = 1
            \end{cases}.
    \end{align*}
    Thus, $\widetilde{PS}(x, \omega)\in [0,B]$ for the first and fourth cases.  For the third case, because $H(0)+\partial H(0)\pi\le H(\pi)$, 
    \begin{align*}
        &H(0)+\frac{H(\pi)-H(0)}{\pi}\\
        \ge& H(0)+\frac{H(0)+\partial H(0)\pi-H(0)}{\pi}\\
        =& H(0)+\partial H(0)\ge 0,
    \end{align*}
    and because $H(\pi)\le (1-\pi) H(0)+\pi H(1)$
    \begin{align*}
        &H(0)+\frac{H(\pi)-H(0)}{\pi}\\
        \le& H(0)+\frac{(1-\pi) H(0)+\pi H(1)-H(0)}{\pi}\\
        =& H(1)\le B.
    \end{align*}
    The second cases follows similarly.  Therefore, $\tilde{H}$ is also ex-post bounded by $B$.
\end{proof}

\begin{proof}[Proof of \cref{thm:singleton}]

\textbf{ex-ante setting}
Let $H^\star$ be the ex-ante bounded $v$-shaped convex function in \cref{thm:singleton} which is in $\mathcal{B}_{exante}$ by direct computation, and $H$ be an arbitrary convex function $\mathcal{B}_{exante}$.  Since $\mathcal{X} = \{X\}$ is singleton, by \cref{lem:vshape}, there exists a $v$-shaped convex function $\tilde{H}\in \mathcal{B}_{exante}$ with $x_0 = \pi$, $a = \frac{H(\pi)-H(0)}{\pi}$, $b = \frac{H(1)-H(\pi)}{1-\pi}$, and $c = H(\pi)$ where $\Jen[X]{\tilde{H}}\ge \Jen[X]{H}$.  

Now we show $\Jen[X]{H^\star}\ge \Jen[X]{\tilde{H}}$.  Let 
$$C = \tilde{H}(0)(1-\pi)+\tilde{H}(1)\pi-\tilde{H}(\pi)\in [0,B]$$ because $\tilde{H}\in \mathcal{B}_{exante}$.  If $C = 0$, $\tilde{H}$ is a linear function, $0 = \Jen[X]{\tilde{H}}\le \Jen[X]{H^\star}$ by the Jensen's inequality.  If $C>0$, we can convert $\tilde{H}$ to $H^\star$ through an affine transformation: 
$$H^\star(x) = \frac{B}{C}\left(\tilde{H}(x)-\tilde{H}(1)x-\tilde{H}(0)(1-x)\right)+B$$ for all $x\in [0,1]$.  By \cref{lem:affine} and $B\ge C>0$, $\Jen[X]{H^\star} = \frac{B}{C}\Jen[X]{\tilde{H}}\ge \Jen[X]{\tilde{H}}$ which completes the proof.

\textbf{ex-post setting}
Let $H^\star$ be the ex-post bounded $v$-shaped convex function in \cref{thm:singleton} which is in $\mathcal{B}_{expost}$ by direct computation, and $H$ be an arbitrary convex function in $\mathcal{B}_{expost}$.  By \cref{lem:vshape}, there exists a $v$-shaped convex function $\tilde{H}\in \mathcal{B}_{expost}$ with $x_0 = \pi$, $a = \frac{H(\pi)-H(0)}{\pi}$, $b = \frac{H(1)-H(\pi)}{1-\pi}$, and $c = H(\pi)$ where $\Jen[X]{\tilde{H}}\ge \Jen[X]{H}$.

Now we show $\Jen[X]{H^\star}\ge \Jen[X]{\tilde{H}}$.  Because $\tilde{H}\in \mathcal{B}_{expost}$, for the score of outcome $\omega = 1$, $H(0)+\partial H(0)$ and $H(1)$ are in $[0,B]$ by \cref{lem:bounded}, and $H(1)\ge H(0)+\partial H(0)$ by convexity.  Thus, 
$H(1)-H(0)-\partial H(0) = (b-a)(1-x_0)\in [0,B].$
Similarly, for the score of outcome $\omega = 0$, we have
$H(0)-H(1)+\partial H(1) = (b-a)x_0\in [0,B]$.  Therefore, 
$$0\le {b-a}\le \frac{B}{\max \pi, 1-\pi}.$$
If $b = a$, $\tilde{H}$ is a positive affine function, and $\Jen[X]{H^\star}\ge \Jen[X]{\tilde{H}} = 0$.  Otherwise, $L:= \frac{B}{\max (\pi, 1-\pi)(b-a)}\ge 1$ and we can convert $\tilde{H}$ to $H^\star$ through an affine transformation:
$$H^\star(x) = L\left(\tilde{H}(x)-\frac{a+b}{2}x-\tilde{H}(\pi)+\frac{a+b}{2}\pi\right)+\frac{B}{2},$$
for all $x\in [0,1]$.
Specifically, the above function is $v$-shaped with $x_0' = \pi$, $a' = \frac{-B}{2\max \pi, 1-\pi}$, $b' = \frac{B}{2\max \pi, 1-\pi}$, and $c' = \frac{B}{2}$.  
By \cref{lem:affine} and $L\ge 1$, $\Jen[X]{H^\star} = L\Jen[X]{\tilde{H}}\ge \Jen[X]{\tilde{H}}$.
\end{proof}

\section{Proof of Theorem~\ref{thm:finite}}
The idea is to construct a linear program whose variables contain the evaluations of $H$ in $\overline{\supp}(\X)$.  To formulate this, we introduce some notations.
Given $|\X| = n$, we set $\X = \{X_i: i = 1, \ldots, n\}$.  For each $i\in [n]$, let the support of $X_i$ be $\supp(X_i) = \{x_{i,j}:j\in [m_i]\}$ with size $|\supp(X_i)| = m_i$, the expectation be $x_{i,0} = \E X_i$, and $\Pr(X_i = x_{i,j}) = p_{i,j}$.  We set $\overline{\supp}(X_i) = \{x_{i,j}: j = 0, \ldots, m_i\}$ and $\overline{\supp}(\X) = \{x_{i,j}:i\in [n], j = 0, \ldots, m_i\}$.  Let $m= \sum_i (1+m_i)$ be the size of $\overline{\supp}(\X)$ and is less than $|\X|\cdot|\mathcal{S}|$.  We further use $\mathcal{A} = \{(i,j):i\in [n], j = 0, \ldots, m_i\}$ to denote the set of indices.  Finally, we set the vertices of the probability simplex $[0,1]$ be $x_{0} = 0$ and $x_1 = 1$, and $\bar{\mathcal{A}} :=\mathcal{A}\cup \{x_0, x_1\}$.  To simplify the notations, we assume $\overline{\supp}(\X)$ does not contain $0$ or $1$, and for all distinct $\alpha, \alpha'$ in $\mathcal{A}$, $x_{\alpha}\neq x_{\alpha'}$.\footnote{Otherwise, we just need to add some equality constraints.  For instance, if $x_\alpha = x_{\alpha'}$, we need to set $h_\alpha = h_{\alpha'}$ as a  constraint.}

Note that the objective value only depends on a finite number of values.  Specifically, let $\obj[\X]{H} := \min_{X\in \mathcal{X}}  \Jen[X]{H}$.  Given $H(x_{\alpha}) = h_\alpha$ for any $\alpha\in \mathcal{A}$, the objective, \cref{eq:obj}, is 
\begin{equation}\label{eq:bin2}
    \obj[\X]{H}  = \min_{i\in [n]} \sum_{k = 1}^{m_i}p_{i,j}h_{i,j}-h_{i,0}.
\end{equation}
Thus, we can first decide $h_\alpha$ to maximize \cref{eq:bin2}, and ``connect'' those points $(x_\alpha, h_\alpha)$ to construct a piece-wise linear function.   To ensure the resulting function is convex, we further require there exists a supporting hyperplane for each $(x_\alpha, h_\alpha)$--- for each $\alpha$ there exists $g_{\alpha}\in \R$ such that $h_{\alpha'}\ge h_{\alpha}+ g_{\alpha}^\top(x_{\alpha'}-x_{\alpha})$ for all $\alpha'\neq \alpha$.

In summary, we set the convex function to be 
\begin{equation}\label{eq:opt_bin}
    H(x) = \max\left\{\max_{\alpha\in \bar{\mathcal{A}}} h_{\alpha}+g_{\alpha}(x-x_\alpha), \min_\alpha h_\alpha\right\},
\end{equation}

\textbf{ex-ante setting} Now in the ex-ante bounded setting with $B = 1$, by \cref{lem:bounded} the collection of $h_{\alpha}$ and $g_{\alpha}$ is a solution of the following linear program,
\begin{equation}\label{eq:lp_bin}
\begin{aligned}
& \max && \min_{i\in [n]} \sum_{k = 1}^{m_i}p_{i,j}h_{i,j}-h_{i,0},\\
& \text{subject to} && h_{\alpha}\in [0,1], &\forall \alpha\in \bar{\mathcal{A}},\\
&&& h_{\alpha'}\ge h_{\alpha}+g_{\alpha}^\top(x_{\alpha'}-x_\alpha),&\forall \alpha, \alpha' \in  \bar{\mathcal{A}}.
\end{aligned}
\end{equation}
The above linear program has $2|\bar{\mathcal{A}}| = O(m)$ variables and $|\bar{\mathcal{A}}|+|\bar{\mathcal{A}}|^2 = O(m^2)$ constraints, so we can solve it in polynomial time with respect to $m$.  

Now we need to show $H$ is convex, bounded in $[0,1]$, and optimal.  It is easy to see for all $\alpha\in \bar{\mathcal{A}}$,
\begin{equation}\label{eq:bin1}
    H(x_{\alpha}) = h_{\alpha}, 
\end{equation}
because 
\begin{align*}
    &H(x_{\alpha})\\
    =& \max\left\{h_{\alpha}, \max_{\alpha'\neq \alpha} h_{\alpha'}+g_{\alpha'}^\top(x_{\alpha}-x_{\alpha'}), \min_{\alpha'} h_{\alpha'}\right\}\tag{$h_{\alpha}+g_{\alpha}^\top(x_{\alpha}-x_{\alpha}) = h_{\alpha}$}\\
    =& \max\left\{h_{\alpha}, \max_{\alpha'\neq \alpha} h_{\alpha'}+g_{\alpha'}^\top(x_\alpha-x_{\alpha'})\right\}\tag{$h_{\alpha}\ge \min_{\alpha'} h_{\alpha'}$}\\
    =& h_{\alpha}\tag{by the constraints in \cref{eq:lp_bin}}
\end{align*}
First because $H$ is the maximum of a collection of linear functions, $H$ is convex.   Second, for the lower bound, by the constraints in \cref{eq:lp_bin} $\min h_\alpha\ge0$ so $0\le \min h_\alpha\le H(x)$ due to \cref{eq:opt_bin}.  For the upper bound, because $H$ is convex, for all $x\in [0,1]$, $H(x)\le \max_k \{H(x_k)\} = \max_k \{h_{k}\}\le 1$ by  \cref{eq:bin1,eq:lp_bin}.  
Finally, for any bounded convex function $\tilde{H}\in \mathcal{H}$, we set $\tilde{h}_\alpha = \tilde{H}(x_\alpha)$ for $\alpha\in \bar{\mathcal{A}}$.  At each $x_\alpha$ we can find a vector $\tilde{g}_\alpha$ such that 
$\tilde{H}(x)\ge \tilde{H}(x_\alpha)+\tilde{g}_\alpha^\top(x-x_\alpha)$
for all $x\in \Delta_d$.\footnote{Specifically, we can construct $\tilde{g}_\alpha$ by finding a support hyperplane to the epigraph of $\tilde{H}$ at $(x_\alpha, h_{\alpha})$, and the vector $\tilde{g}_\alpha$ is called subgradient.}  Since $\tilde{H}$ is convex and in $\mathcal{H}$, the collection of $\tilde{h}_\alpha$ and $\tilde{g}_\alpha$ is a feasible solution to \cref{eq:lp_bin}, and $\obj[\X]{\tilde{H}}\le \obj[\X]{H}$.

\textbf{ex-post setting}  The ex-post setting is identical except the first bounded condition\begin{equation}\label{eq:lp_bin_post}
\begin{aligned}
& \max && \min_{i\in [n]} \sum_{k = 1}^{m_i}p_{i,j}h_{i,j}-h_{i,0},\\
& \text{s.t.} && h_{\alpha}+ g_{\alpha}^\top(1-x_{\alpha}), h_{\alpha}-g_{\alpha}^\top x_{\alpha}\in [0,1], &\forall \alpha\in \bar{\mathcal{A}},\\
&&& h_{\alpha'}\ge h_{\alpha}+g_{\alpha}^\top(x_{\alpha'}-x_\alpha),&\forall \alpha, \alpha' \in  \bar{\mathcal{A}}.
\end{aligned}
\end{equation}
The above linear program has $O(m)$ variables and $4|\bar{\mathcal{A}}|+|\bar{\mathcal{A}}|^2 = O(m^2)$ constraints which can be solve efficiently.  Finally, using a similar argument above we can show that $H$ is convex, ex-post bounded and optimal.  

\section{Proofs and Details for Section~\ref{sec:infinite}}


\subsection{Proofs for Lemmas~\ref{lem:lipschitz} and \ref{lem:covering}}
\begin{proof}[Proof of \cref{lem:lipschitz}]
    We first show a bounded $H$ is Lipschitz in the support of $X$ and $X'$.  First, since $H$ is convex, for all $x\neq x'$, $H(x)\ge H(x')+\partial H(x')(x-x')$ and $H(x')\ge H(x)+\partial H(x)(x'-x)$.  Hence, $\partial H(x)(x'-x)\le H(x')-H(x)\le \partial H(x')(x'-x)$, and 
    $$|H(x')-H(x)|\le \max(|\partial H(x')|, |\partial H(x)|)|(x'-x)|.$$
    That is, the maximum norm of subgradients can bound the Lipschitz constant.  

    If $H\in \mathcal{B}_{expost}$, taking the difference between two terms in \cref{lem:bounded}, we have 
    $|\partial H(x)|\le B$ for all $x$, so 
    $|H(x')-H(x)|\le \max_{x\in [0,1]}|\partial H(x)||x'-x| \le B|x'-x|.$  Thus, 
    \begin{align*}
        &|\Jen[X]{H}-\Jen[X']{H}|\\
        \le& |\E_X H(X)-\E_{X'} H(X')|+ |H(\E X)-H(\E X')|\\
        \le&  B d_{EM}(X,X')+ B|\E X-\E X'|\\
        \le& 2Bd_{EM}(X,X').
    \end{align*}

    On the other hand, if $H\in \mathcal{B}_{exante}$, by \cref{lem:bounded}, we have $H(1)-H(x)$ and $H(0)-H(x)$ are less than $B$ for all $x$.  Since $H$ is convex, $\partial H(x)(1-x)\le H(1)-H(x)\le B$ and $\partial H(x)(-x)\le H(0)-H(x)\le B$.  If $x\in [1/L, 1-1/L]$, $\partial H(x)\in [-B/x, B/(1-x)]\in [-LB, LB]$.  The rest of the proof follows the ex-post setting.
\end{proof}

\begin{proof}[Proof of \cref{lem:covering}]
    For any $\pi\in [\delta, 1-\delta]$, there exists $\pi'$ so that $(\pi', \sigma)\in \mathcal{P}_{\sigma, N}$ and $|\pi-\pi'|\le 1/(2N)$.  Then the TVD between $P = (\pi, \sigma)$ and $P' = (\pi', \sigma)$ is 
    \begin{align*}
        &\frac{1}{2}\sum_{w,s} |P[w, s]-P'[w,s]|\\
        =& \frac{1}{2}\sum_s |\pi \sigma(s|1)-\pi'\sigma(s|1)|+|(1-\pi) \sigma(s|0)-(1-\pi')\sigma(s|0)|\\
        =& \frac{1}{2}|\pi-\pi'|\sum_s \sigma(s|1)+\sigma(s|0)\\
        =& |\pi-\pi'|\le \frac{1}{2N}
    \end{align*}

    Now we consider $\rho$-correlated information structures.  First, given $(\pi, \sigma)\in \mathcal{P}_\rho$, 
    \begin{align*}
        P(W = 1, S = 1) =& \pi(\rho+(1-\rho)\pi)\\
        P(W = 1, S = 0) =& \pi(1-\rho)(1-\pi)\\
        P(W = 0, S = 1) =& (1-\pi)(1-\rho)\pi\\
        P(W = 1, S = 0) =& (1-\pi)(\rho+(1-\rho)(1-\pi))
    \end{align*}
    Similar, there exists $(\pi', \sigma')\in \mathcal{P}_{\rho, N}$ with $|\pi-\pi'|\le \frac{1}{2N}$ and the TVD between them is the sum of the following four terms $|P[w, s]-P'[w,s]|$ with $w, s = 0,1$ divided by $2$.
    First,
    \begin{align*}
        &|P(W = 1, S = 1)-P'(W = 1, S = 1)|\\
        =& |\pi(\rho+(1-\rho)\pi)-\pi'(\rho+(1-\rho)\pi')|\\
        \le& \rho|\pi-\pi'|+(1-\rho)|\pi^2-(\pi')^2|
    \end{align*}
    Second,
    \begin{align*}
        &|P(W = 0, S = 0)-P'(W = 0, S = 0)|\\
        =& |(1-\pi)(\rho+(1-\rho)(1-\pi))-(1-\pi')(\rho+(1-\rho)(1-\pi'))|\\
        \le& \rho|\pi-\pi'|+(1-\rho)|(1-\pi)^2-(1-\pi')^2|
    \end{align*}
    Finally, $|P(W = 1, S = 0)-P'(W = 1, S = 0)| = |P(W = 0, S = 1)-P'(W = 0, S = 1)|$ and 
    \begin{align*}
        &|P(W = 1, S = 0)-P'(W = 1, S = 0)|\\
        =& |\pi(1-\rho)(1-\pi)-\pi'(1-\rho)(1-\pi')|\\
        =& (1-\rho)|\pi(1-\pi)-\pi'(1-\pi')|\\
        \le& (1-\rho)|\pi-\pi'|\tag{by Taylor expension}
    \end{align*}
    Combining these three, we have
    
    \begin{align*}
        &\frac{1}{2}\sum_{w,s} |P[w, s]-P'[w,s]|\\
        \le& \frac{1}{2}\rho|\pi-\pi'|+\frac{1}{2}(1-\rho)|\pi^2-(\pi')^2|\\
        &+\frac{1}{2}\rho|\pi-\pi'|+\frac{1}{2}(1-\rho)|(1-\pi)^2-(1-\pi')^2|\\
        &+(1-\rho)|\pi-\pi'|\\
        =& |\pi-\pi'|+(1-\rho)|\pi-\pi'|\\
        =& \frac{2-\rho}{2N}
    \end{align*}
\end{proof}
\subsection{EMD, TVD, and Couplings}
This section introduces couplings which can be skipped for familiar readers.  For any two distributions $P$ and $Q$ on a common domain $\Omega$ with a metric $d:\Omega^2\to \R$, let $\Gamma(P,Q)$ denote the set of all joint distributions on $\Omega\times \Omega$ with marginals $P$ and $Q$. A joint distribution $\gamma\in \Gamma(P,Q)$ is called a \emph{coupling} or \emph{transportation plan} between $P$ and $Q$.

The \emph{Wasserstein distance} between $P$ and $Q$ is $\inf_{\gamma \in \Gamma(P,Q)}\E_{(x, y)\sim \gamma}[d(x, y)]$.  Given a real-valued function $f$ on $\Omega$, 
$$\|f\|_{L, d}:= \sup_{x\neq y\in \Omega}|f(x)-f(y)|/d(x, y).$$ Note that when $\Omega = [0,1]$ and $d = |\cdot|$ is absolute value, $\|f\|_{L,d} = \|f\|_{Lip}$.

Now we state the Kantorovich-Rubinstein Theorems (Theorem 11.8.2~\citep{dudley2018real}) which shows the duality between Wasserstein distance and earth mover's distance/total variation distance: 
    For any metric $d$ and two distributions $P, Q$ on $\Omega$, 
    \begin{equation}\label{eq:dual}
        \sup_{f: \|f\|_{L,d} = 1} \E_{x\sim P} f(x)-\E_{y\sim Q}f(y) = \inf_{\gamma\in \Gamma(P, Q)} \E_{(x, y)\sim \gamma}[d(x,y)].
    \end{equation}
In particular, by taking $d$ as $1$-norm, 
\begin{equation}\label{eq:dual_emd}
    d_{EM}(P, Q) = \inf_{\gamma\in \Gamma(P, Q)} \E_{(x, y)\sim \gamma}[|x-y|].
\end{equation}

Taking $d$ as the discrete metric where $d(x,y) = 1$ for all $x\neq y$, 
$$d_{TV}(P,Q) = \inf_{\gamma\in \Gamma(P, Q)} \E_{(x, y)\sim \gamma}[\mathbf{1}[x\neq y]].$$  

\subsection{Proof of Lemma~\ref{prop:tv2emd}}

\begin{proof}[Proof of \Cref{prop:tv2emd}]
We will use the maximal coupling between information structures on the signal space to bound the EMD between posterior predictions.  Specifically, we sample $S = s$ with probability $P[S = s]$ and output $\hat{X} = P[W = 1\mid S = s]$.  Then, let $\hat{X}' = P'[W = 1\mid S = s]$ with probability $\min(1, P'[S = s]/P[S = s])$ and any other feasible value  satisfying the marginal distribution: $\hat{X} = X$ and $\hat{X}' = X'$ in distribution.

Let $d_{TV}(P, P') = \delta$.  We define the matching event as
    $\mathcal{E}:=\{(\hat{X}, \hat{X}'): \exists s\in \mathcal{S}, \hat{X} = P[W = 1\mid S = s], \hat{X}' = P'[W = 1\mid S = s]\}.$ Then the probability of not matching is 
    \begin{equation}\label{eq:tv2emd1}
        \begin{aligned}
        &1-\sum_{s}P[S = s]\min\left(1, \frac{P'[S = s]}{P[S = s]}\right)\\
        =& \frac{1}{2}\sum_{s}\left|P[S = s]-P'[S = s]\right|\\
        \le& d_{TV}(P, P') = \delta
    \end{aligned}
    \end{equation}
    By the dual form of EMD in~\Cref{eq:dual}, the above coupling satisfies
    \begin{align*}
        &d_{EM}(X,X')\\
        \le& \E_\gamma\left[|\hat{X}-\hat{X}'|\right]\\
        =& \E_\gamma\left[|\hat{X}-\hat{X}'|\mid \neg \mathcal{E}\right]\Pr[\neg \mathcal{E}]+\E_\gamma\left[|\hat{X}-\hat{X}'|\mid \mathcal{E}\right]\Pr[\mathcal{E}].
    \end{align*}
    We bound these two terms separately. First, because the difference between predictions is always bounded by $1$, with \Cref{eq:tv2emd1}
    $$\E\left[|\hat{X}-\hat{X}'|\mid \neg \mathcal{E}\right]\Pr[\neg \mathcal{E}]\le \max_{x, x'\in [0,1]}|x-x'|\Pr[\neg \mathcal{E}]\le \delta.$$
    Now we bound the second term.  Given $s\in \mathcal{S}$, we set $\delta(s) = \sum_w|P[W = w,S = s]-P'[W = w,S = s]|$ and write $P[W = 1, S = s] = P[1, s]$ and $P[S = s] = P[s]$. 
 We have
    \begin{align*}
        &|P[W = 1|S = s]-P'[W = 1|S = s]|\\
        =& \frac{1}{P[s]P'[s]}\left|P[1, s]P'[s]-P'[1, s]P[s]\right|\\
        =& \frac{1}{P[s]P'[s]}\left|P[1, s]P'[0,s]-P[0, s]P'[1, s]\right|\\
        \le& \frac{1}{P[s]P'[s]}\big(\left|P[1, s]P'[0,s]-P[1, s]P[0,s]\right|\\
        &+\left|P[1, s]P[0,s]-P[0, s]P'[1, s]\right|\big)\\
        \le & \frac{1}{P[s]P'[s]}\left(P[1, s]\delta(s)+P[0,s]\delta(s)\right)\\
        =& \frac{P[s]\delta(s)}{P[s]P'[s]}.
    \end{align*}
    Using a similar argument, we have $|P[W = 1|s]-P'[W = 1|s]|\le \frac{\delta(s)}{\max P[s], P'[s]}$.  Therefore, we have
    \begin{align*}
         &\E\left[|\hat{X}-\hat{X}'|\mid \mathcal{E}\right]\Pr[\mathcal{E}]\\
        =& \sum_{s}P[s]\frac{\min(P[s], P'[s])}{P[s]}\left|P[W = 1|s]-P'[W = 1|s]\right|\\
         =& \sum_{s}\min(P[s], P'[s])\left|P[W = 1|s]-P'[W = 1|s]\right|\\
         \le& \sum_{s}\min(P[s], P'[s])\frac{\delta(s)}{(\max P[s], P'[s])}\\
         \le& \sum_{s}\delta(s) = \delta
    \end{align*}
    That completes the proof.
\end{proof}

\section{Additional Simulations}
This section includes additional simulation for $\rho$-correlated model on ex-post model, and information structures with homogeneous experiments.
\subsection[Beta-Bernoulli information structures in the ex-post setting]{$\rho$-correlated information structures in the ex-post setting}

First we plot the associated convex functions for ex-post bounded setting analogous to \cref{sec:sim} in the top of \cref{fig:beta_post}.  We use 1) a quadratic scoring rules for the ex post bounded case with an associated convex function $(x-1/2)^2+3/4$, 2) a $v$-shaped scoring rule with $(a, b, c, x0) = (-1, 1, 1/2, 1/2)$, and the piecewise linear scoring rule for $\mathcal{P}_{\rho, N}$ with $N = 50$ and $\delta = 0.05$ in the ex-post setting with $\rho = 0.25$ and $0.025$.  Note that the log scoring rule cannot be ex-post bounded. The bottom two plots in \cref{fig:beta_post} are the information gains under these three functions on $\mathcal{P}_{\rho, N'}$ with $N' = 1000$, $\rho = 0.25, 0.025$, and $\delta = 0.05$. 


\begin{figure}
\centering
    \begin{minipage}[t]{.32\textwidth}
        \includegraphics[width=\textwidth]{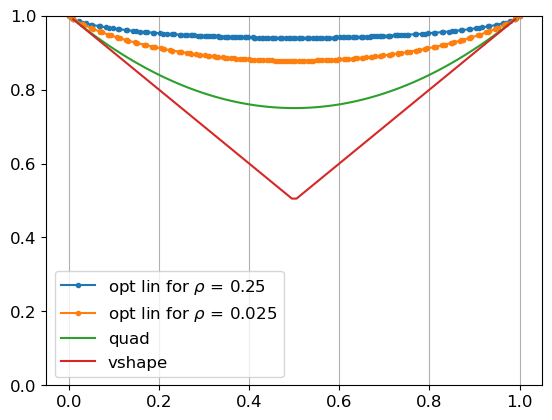}
    \end{minipage}
    \hfill
    \begin{minipage}[t]{.32\textwidth}
        \includegraphics[width=\textwidth]{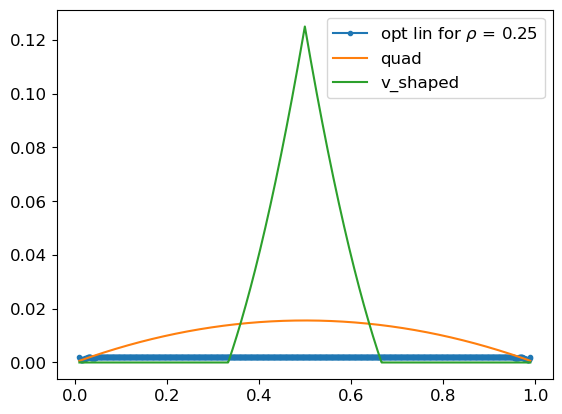}
    \end{minipage}
    \hfill
    \begin{minipage}[t]{.32\textwidth}
        \includegraphics[width=\textwidth]{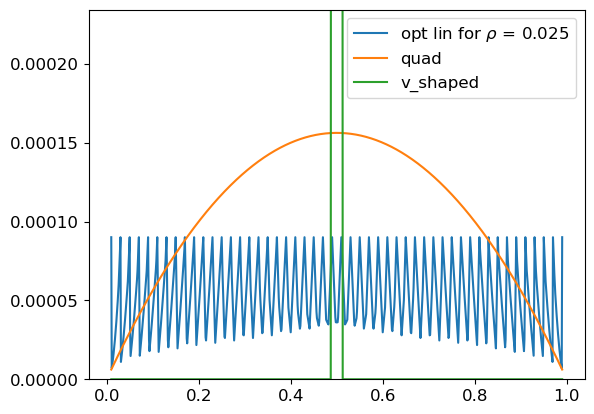}
    \end{minipage}  \caption{The top plot is the associated convex functions for $\rho$-correlated information structures in the ex-post bounded setting with $B = 1$.  The middle one is the information gain when $\rho = 0.25$, and the bottom one is the information gain when $\rho = 0.025$.
    }\label{fig:beta_post}
\end{figure}


\begin{figure}
\centering
    \begin{minipage}[t]{.32\textwidth}
        \includegraphics[width=\textwidth]{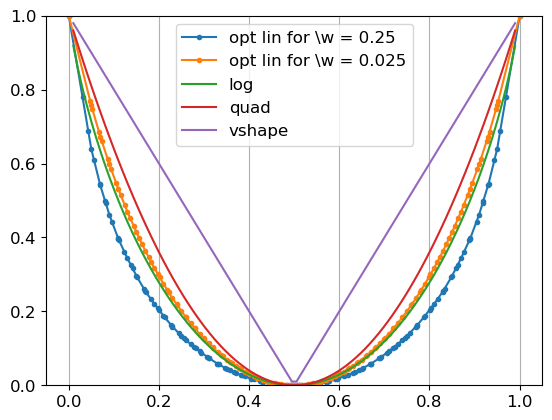}
    \end{minipage}
    \hfill
    \begin{minipage}[t]{.32\textwidth}
        \includegraphics[width=\textwidth]{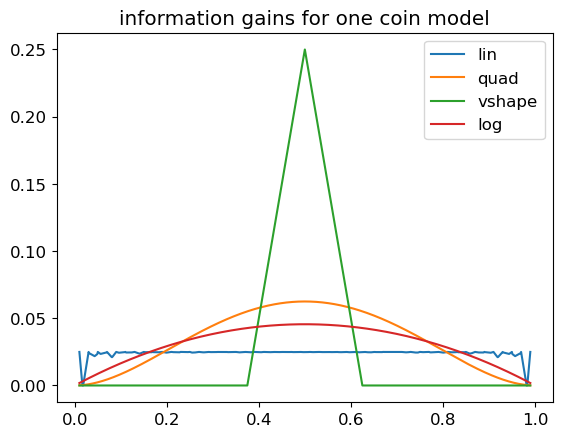}
    \end{minipage}
    \hfill
    \begin{minipage}[t]{.32\textwidth}
        \includegraphics[width=\textwidth]{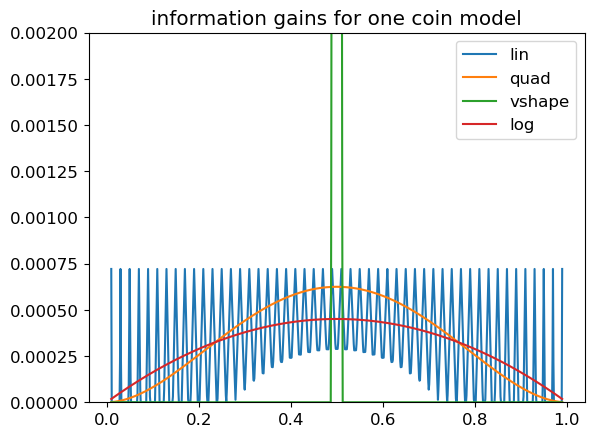}
    \end{minipage}
    \caption{The top plot is the associated convex functions for information structures with homogeneous one coin model in the ex-ante bounded setting with $B = 1$.  The middle one is the information gain when $\rho = 0.25$, and the bottom one is the information gain when $\rho = 0.025$.
    }\label{fig:coin_ante}
\end{figure}


\subsection{Information structures with homogeneous experiment}
Now, we test scoring rules' performance on information structures with homogeneous experiment.  We use a simple one coin model with parameter $\xi\in [0,1]$ so that the signal space is binary and the experiment satisfies $\sigma(1|1) = \sigma(0|0) = \frac{1+\xi}{2}$ and $\frac{1-\xi}{2}$ otherwise.  Similar to $\rho$-correlated information structures, the value of $\xi$ control the correlation between the state and signal. 

The top plot in \cref{fig:coin_ante} are the associated convex functions for ex-ante bounded setting analogous to \cref{sec:sim}.  Note that similar to \cref{fig:func2}, the optimal solution for the large correlation setting ($\xi = 0.25$) is more curved at the boundary instead of center. (more $u$-shaped than $v$-shaped).  Intuitively, this is due to the movement of posterior is relatively smaller at the boundary as the correlation is increased.  Then bottom two plot in \cref{fig:coin_ante} are the information gains of these three functions on $\mathcal{P}_{\rho, N'}$ with $N = 1000$, $\delta = 0.05$, and $\xi = 0.25, 0.025$. Then \cref{fig:coin_post} presents results for the ex-post setting.



\begin{figure}
\centering
    \begin{minipage}[t]{.32\textwidth}
        \includegraphics[width=\textwidth]{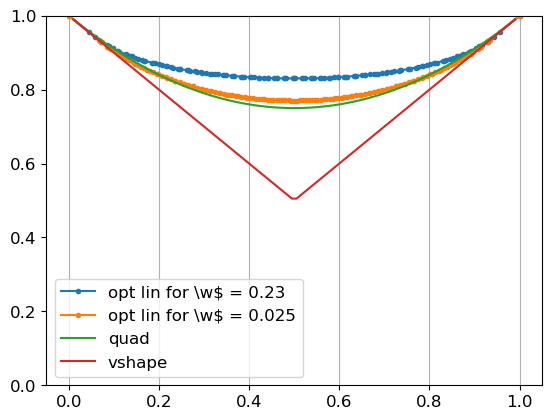}
    \end{minipage}
    \hfill
    \begin{minipage}[t]{.32\textwidth}
        \includegraphics[width=\textwidth]{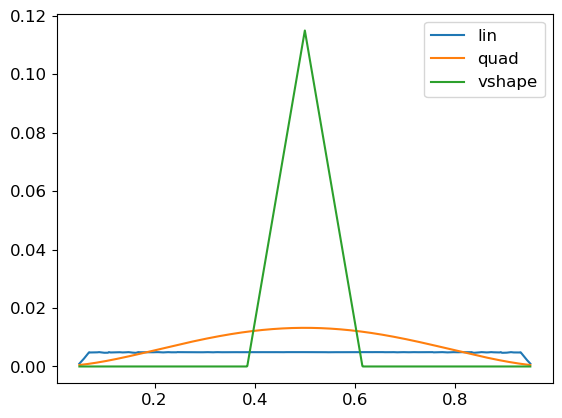}
    \end{minipage}
    \hfill
    \begin{minipage}[t]{.32\textwidth}
        \includegraphics[width=\textwidth]{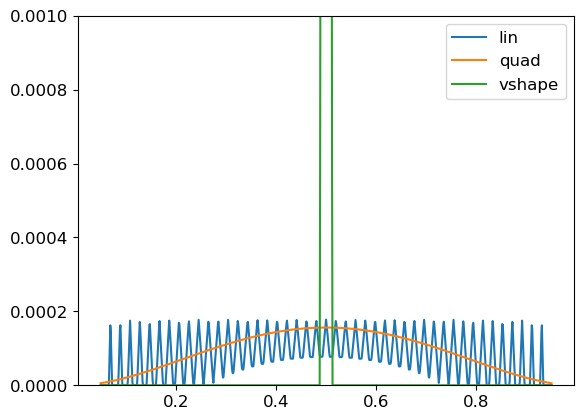}
    \end{minipage}
    \caption{The top plot is the associated convex functions for information structures with homogeneous one coin model for ex-post bounded setting with $B = 1$.  The middle one is the information gain when $\rho = 0.25$, and the bottom one is the information gain when $\rho = 0.025$.
    }\label{fig:coin_post}
\end{figure}

\section[General State Space]{$d$-dimensional State Space}
In this section, we demonstrate how to generalize our result to general categorical state space $\Omega$ with $d$ possible outcome.  For simplicity, we only consider ex-ante bounded setting.  
\subsection{Notations}
Here we list some notations.  Given positive integer $d$ and $n$, let $[n] := \{1, \ldots, n\}$, $\Delta = \Delta_{d} = \Delta(\Omega)\subset \R^d$ is probability simplex over $\Omega = [d]$.  The vertices of the simplex $\Delta_d$ are  $\hat{e}_j$ for $j\in [d]$ which is also the standard basis of $\R^d$.
We call all one vector $\mathbf{1} := (1, \ldots, 1)\in\R^d$ and $c := \frac{1}{d}\mathbf{1}$.

\paragraph{Proper Scoring Rules}
A \emph{scoring rule} for a random variable $W\in \Omega$ is a function $PS:\Omega\times\Delta(\Omega)\to \R$ where $PS(\omega, \hat{x})$ is the score assigned to a prediction $\hat{x}\in \Delta(\Omega)$ when $W = \omega$.  The scoring rule is (strict) proper if for all random variable $W$ on $\Omega$ with distribution $x$, setting $\hat{x} = x$ (uniquely) maximizes the expected score $\E_{W\sim x} PS(W,\hat{x})$.  In other words, if $W$ is distributed according to $x$, then truthfully reporting $x$ can maximizes the expected score.

\begin{theorem}[Savage representation~\cite{mccarthy1956measures,savage1971elicitation,gneiting2007strictly}]
For every (strict) proper scoring rule $PS$, there exists a (strictly) convex function $H:\Delta(\Omega)\to \R$ so that for all $\omega\in \Omega$ and $x\in \Delta(\Omega)$
$$PS(\omega, x) = H(x)+ \partial H(x)\cdot(\mathbf{1}_\omega-x)$$
where $\partial H(x)$ is a sub-gradient of $H$ at $x$ and $\mathbf{1}_\omega\in \Delta(\Omega)$ is the indicator function with $\mathbf{1}_\omega(w') = 1$ if $w'  = \omega$ and $0$ otherwise.

Conversely, for every (strictly) convex function $H: \Delta(\Omega)\to \R$, there exists a (strict) proper scoring rule such that the above condition hold.
\end{theorem}

We list some common proper scoring rules with associated convex functions which are scaled to be in $\mathcal{H}$:
1) quadratic scoring rule $Q(x) = \frac{d}{d-1}\sum_{k = 1}^d (x_k^2-1/d)^2$, 2) spherical scoring rule $H_s(x) = \frac{\sqrt{d}}{\sqrt{d}-1}\sqrt{\sum x_k^2}-\frac{1}{\sqrt{d}-1}$, 3) log scoring rule $H_{\ln}(x) = \frac{1}{\ln d}\sum_k x_k\ln x_k+1$.

When the state space $\Omega$ is binary, the associated convex function is one-dimensional.  For every proper scoring rule $PS$, there exists a convex function $H:[0,1]\to\R$ so that for all $x\in [0,1]$ and binary event $\omega\in \{0,1\}$
\begin{equation}\label{eq:ps_bin}
    PS(\omega, x) = \begin{cases}
H(x)+ \partial H(x)\cdot (1-x)&\text{ if } \omega = 1,\\
H(x)-\partial H(x)\cdot x&\text{ if } \omega = 0.\end{cases}
\end{equation}

\subsection{Singleton Information Structure}
As a warm-up, let's consider the principal exactly knows the agent's information structure so that $\X = \{X\}$ is a singleton.  We show the optimal $H$ can be an upside down pyramid. (\cref{fig:multi_v_shape})  

\begin{theorem}[name = singleton, label = thm:multi_singleton]
If $\X = \{X\}$ is singleton and the state space $\Omega = [d]$, there exists an optimal scoring rule associated with an upside down pyramid $H^\star$ such that the epigraph of $H^\star$ is the convex hull with vertices $(\E X,0)$, and $(\hat{e}_k, 1)$ for all $k\in [d]$.
\end{theorem}
Since the information gain $\Jen[X]{H}$ is an integration of $H(x)-H(\E X)$ over all $x\in \Delta_d$ for any $H$, we can prove \cref{thm:singleton} by a pointwise inequality, $H(x)-H(\E X)\le H^\star(x)-H^\star(\E X)$ for all $x$.  See appendix for more details. 

When the state space is binary $\Omega = \{0,1\}$, by \cref{eq:ps_bin}, the upside down pyramid $H$ is \emph{v-shaped}, and \cref{thm:singleton} yields the following corollary.
\begin{corollary}\label{prop:v-shape}
If the state space $\Omega = \{0,1\}$ and $\X = \{X\}$, there exists an optimal scoring rule associated with a v-shape $H$ such that
$$H(x) = \begin{cases}
\frac{-1}{\E X}(x-\E X) &\text{ if } x<\E X\\
\frac{1}{1-\E X}(x-\E X)&\text{ if } x\ge \E X\end{cases}$$
\end{corollary}
These results suggest the principal should choose $H$ that is ``curved'' at the prior in order to incentivize the agent to derive the signal and move away from the prior.  This intuition is useful for the later sections.

\begin{figure}
    \centering
    \begin{tabular}{cc}
\includegraphics[width=0.4\textwidth]{images/v_shape.png}& \includegraphics[width=0.5\textwidth]{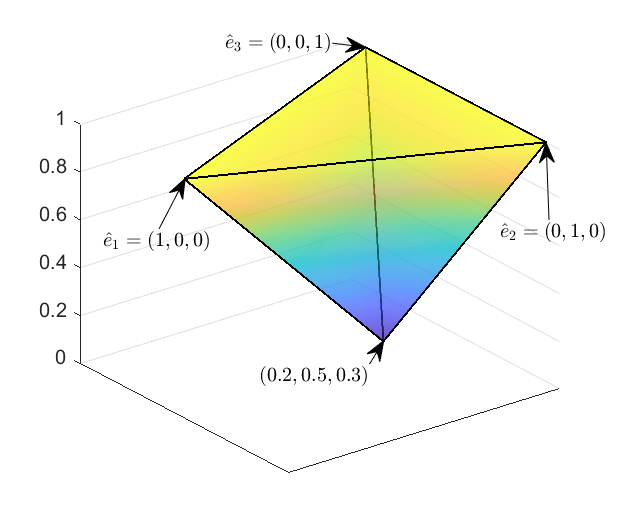}
\end{tabular}
    \caption{The top panel shows several examples of the optimal $H$ in \cref{prop:v-shape} which have $\E X = 0.2, 0.5$, and $0.7$ respectively.  The bottom panel shows the optimal $H^\star$ in \cref{thm:singleton} when $d = 3$, $\X = \{X\}$, and $\Pr[X = 1] = 0.2,\Pr[X = 2] = 0.5, \Pr[X = 3] = 0.3$.}
    \label{fig:multi_v_shape}
\end{figure}
\begin{proof}[Proof of \cref{thm:multi_singleton}]
Let $H^\star$ be the upside down pyramid in \cref{thm:singleton}, and $H\in \mathcal{H}$ be an arbitrary bounded convex function.  Since $\X = \{X\}$ is singleton, it is sufficient to prove $\Jen[X]{H}\le \Jen[X]{H^\star}$.  

Let $h_0 := H(\E X)$.  If $h_0 = 1$, $H$ is a constant function and $0 = \Jen[X]{H}\le \Jen[X]{H^\star}$ by the Jensen's inequality.  Now we consider $h_0<1$.  Let $\tilde{H}$ be a convex piecewise linear function whose epigraph has vertices at $(\E X, h_0)$, and $(\hat{e}_k, 1)$ for all $k\in [d]$.   First, because $H\in \mathcal{H}$, $H(\hat{e}_k)\le \tilde{H}(\hat{e}_k) = 1$ for all $k\in [d]$, and $H(\E X) = \tilde{H}(\E X) = h_0$.  Then, the epigraph of $\tilde{H}$ is contained in the epigraph of $H$, so
\begin{equation}\label{eq:singleton1}
     H(x)-H(\E X)\le \tilde{H}(x)-\tilde{H}(\E X) \text{ for all }x\in \Delta_d.
\end{equation}
Second, because the epigraphs of $H^\star$ and $\tilde{H}$ are both upside down pyramids, and the vertices are aligned, we can convert $\tilde{H}$ to $H^\star$ through an affine transformation: $H(x) = \frac{\tilde{H}(x)-h_0}{1-h_0}$ for all $x\in \Delta_d$.  Therefore, for all $x\in \Delta_d$,
\begin{equation}\label{eq:singleton2}
    \tilde{H}(x)-\tilde{H}(\E X) = (1-h_0) \left(H^\star(x)-H^\star(\E X)\right)\le H^\star(x)-H^\star(\E X),
\end{equation}
because $0\le h_0<1$, and both sides are non-negative.  
Combining \cref{eq:singleton1,eq:singleton2}, we have
$\Jen[X]{H} = \E_X\left[H(X)-H(\E X)\right]\le \E_X\left[H^\star(X)-H^\star(\E X)\right] = \Jen[X]{H^\star}$, and complete the proof.
\end{proof}
\subsection{Finite Information Structures}\label{sec:multi_finite}
In this section, we give a polynomial time algorithm that computes an optimal scoring rule when the collection of information structures is finite as defined below.

\begin{definition}\label{def:finte}
We call a collection of information structures $\X$ \emph{finite} if $|\X|$ is finite and all $X\in \X$ has a finite support $|\supp(X)|<\infty$.

When $\X$ is finite, let $\overline{\supp}(X) := \supp(X)\cup \{\E X\}\subset \Delta(\Omega)$ for all $X\in \X$, and $\overline{\supp}(\X) := \cup_{X\in \X} \overline{\supp}(X)$.
\end{definition}
The notion of finite collection of information structures is natural when there is a finite number of heterogeneous agents with finite set of information.  

\begin{theorem}\label{thm:multi_finite}
If the state space $\Omega = [d]$, and $\X$ is finite with $|\overline{\supp}(\X)| = m$, there exists an algorithm that computes an optimal ex-ante bounded proper scoring rule and the running time is polynomial in $d$ and $m$.
\end{theorem}

The main idea is that when $\X$ is finite $\obj[X]{H}$ in \cref{eq:obj} only depends on the evaluations of $H$ in $\overline{\supp}(\X)$.  Thus, instead of searching for all possible bounded scoring rules, we can reduce the dimension of \cref{prob:opt} and use a linear programming whose variables contain the evaluations of $H$ in $\overline{\supp}(\X)$ and add linear constraints to ensure those evaluation can be extended to a convex function.
This observation allows us the solve the problem in weakly polynomial time which is polynomial in $d$ and $m$ but may not be polynomial in the representation size.  We present the formal proof in the appendix.

Note that if $\X$ contains a constant information structure $X$, the resulting linear programming~\eqref{eq:lp_bin} will output an arbitrary piecewise linear convex function, because the objective value is always zero (\cref{fn:multi_equality}).

\begin{proof}[Proof of Theorem~\ref{thm:multi_finite}]
The idea is to construct a linear programming whose variables contain the evaluations of $H$ in $\overline{\supp}(\X)$.  To formulate this, we introduce some notations.
Given $|\X| = n$, we set $\X = \{X_i: i = 1, \ldots, n\}$.  For each $i\in [n]$, let the support of $X_i$ be $\supp(X_i) = \{x_{i,j}:j\in [m_i]\}$ with size $|\supp(X_i)| = m_i$.  Additionally, let the expectation be $x_{i,0} = \E X_i$, and $\Pr(X_i = x_{i,j}) = p_{i,j}$.  Hence $\overline{\supp}(X_i) = \{x_{i,j}: j = 0, \ldots, m_i\}$ and $\overline{\supp}(\X) = \{x_{i,j}:i\in [n], j = 0, \ldots, m_i\}$.  We further use $\mathcal{A} = \{(i,j):i\in [n], j = 0, \ldots, m_i\}$ to denote the set of indices.  Finally, we set the vertices of the probability simplex $\Delta_d$ be $x_{k} = \hat{e}_k$ for all $k\in [d]$, and $\bar{\mathcal{A}} :=\mathcal{A}\cup [d]$.  To simplify the notations, we assume $\overline{\supp}(\X)$ does not contain any vertex of $\Delta_d$ $\hat{e}_k$ for $k = 1, \ldots, d$ and for all distinct $\alpha, \alpha'$ in $\mathcal{A}$, $x_{\alpha}\neq x_{\alpha'}$.\footnote{Otherwise, we just need to add some equality constraints.  For instance, if $x_\alpha = x_{\alpha'}$, we need to set $h_\alpha = h_{\alpha'}$ as a  constraint.\label{fn:multi_equality}}

Note that the objective value only depends on a finite number of values.  Specifically, given $H(x_{\alpha}) = h_\alpha$ for any $\alpha\in \mathcal{A}$, the objective, \cref{eq:obj}, is 
\begin{equation}\label{eq:multi_bin2}
    \obj[\X]{H} = \min_{i\in [n]} \sum_{k = 1}^{m_i}p_{i,j}h_{i,j}-h_{i,0}.
\end{equation}
Thus, we can first decide $h_\alpha$ to maximize \cref{eq:multi_bin2}, and ``connect'' those points $(x_\alpha, h_\alpha)$ to construct a piece-wise linear function.   To ensure the resulting function is convex, we further require there exists a supporting hyperplane for each $(x_\alpha, h_\alpha)$--- for each $\alpha$ there exists $g_{\alpha}\in \R^d$ such that $h_{\alpha'}\ge h_{\alpha}+ g_{\alpha}^\top(x_{\alpha'}-x_{\alpha})$ for all $\alpha'\neq \alpha$.

In summary, we set the convex function to be 
\begin{equation}\label{eq:multi_opt_bin}
    H(x) = \max\left\{\max_{\alpha\in \bar{\mathcal{A}}} h_{\alpha}+g_{\alpha}(x-x_\alpha), \min_\alpha h_\alpha\right\},
\end{equation}
and the collection of $h_{\alpha}$ and $g_{\alpha}$ is a solution of the following linear programming,
\begin{equation}\label{eq:multi_lp_bin}
\begin{aligned}
& \max && \min_{i\in [n]} \sum_{k = 1}^{m_i}p_{i,j}h_{i,j}-h_{i,0},\\
& \text{subject to} && h_{\alpha}\in [0,1], &\forall \alpha\in \bar{\mathcal{A}},\\
&&& h_{\alpha'}\ge h_{\alpha}+g_{\alpha}^\top(x_{\alpha'}-x_\alpha),&\forall \alpha, \alpha' \in  \bar{\mathcal{A}}.
\end{aligned}
\end{equation}
The above linear programming has $(1+d)|\bar{\mathcal{A}}| = O(d(m+d))$ variables and $|\bar{\mathcal{A}}|+|\bar{\mathcal{A}}|^2 = O((m+d)^2)$ constraints, so we can solve it in polynomial time with respect to $m$ and $d$.  

Now we need to show $H$ is 1) convex, 2) bounded in $[0,1]$, and  3) optimal.  It is easy to see for all $\alpha\in \bar{\mathcal{A}}$,
\begin{equation}\label{eq:multi_bin1}
    H(x_{\alpha}) = h_{\alpha}, 
\end{equation}
because 
\begin{align*}
    &H(x_{\alpha})\\
    =& \max\left\{h_{\alpha}, \max_{\alpha'\neq \alpha} h_{\alpha'}+g_{\alpha'}(x_{\alpha}-x_{\alpha'}), \min_{\alpha'} h_{\alpha'}\right\}\tag{$h_{\alpha}+g_{\alpha}(x_{\alpha}-x_{\alpha}) = h_{\alpha}$}\\
    =& \max\left\{h_{\alpha}, \max_{\alpha'\neq \alpha} h_{\alpha'}+g_{\alpha'}(x_\alpha-x_{\alpha'})\right\}\tag{$h_{\alpha}\ge \min_{\alpha'} h_{\alpha'}$}\\
    =& h_{\alpha}\tag{by the constraints in \cref{eq:multi_lp_bin}}
\end{align*}
First because $H$ is the maximum of a collection of linear functions, $H$ is convex.   Second, for the lower bound, by the constraints in \cref{eq:multi_lp_bin} $\min h_\alpha\ge0$ so $0\le \min h_\alpha\le H(x)$ due to \cref{eq:multi_opt_bin}.  For the upper bound, because $H$ is convex, for all $x\in \Delta_d$, $H(x)\le \max_k \{H(\hat{e}_k)\} = \max_k \{h_{k}\}\le 1$ by  \cref{eq:multi_bin1,eq:multi_lp_bin}.  
Finally, for any bounded convex function $\tilde{H}\in \mathcal{H}$, we set $\tilde{h}_\alpha = \tilde{H}(x_\alpha)$ for $\alpha\in \bar{\mathcal{A}}$.  At each $x_\alpha$ we can find a vector $\tilde{g}_\alpha$ such that 
$\tilde{H}(x)\ge \tilde{H}(x_\alpha)+\tilde{g}_\alpha^\top(x-x_\alpha)$
for all $x\in \Delta_d$.\footnote{Specifically, we can construct $\tilde{g}_\alpha$ by finding a support hyperplane to the epigraph of $\tilde{H}$ at $(x_\alpha, h_{\alpha})$, and the vector $\tilde{g}_\alpha$ is called subgradient.}  Since $\tilde{H}$ is convex and in $\mathcal{H}$, the collection of $\tilde{h}_\alpha$ and $\tilde{g}_\alpha$ is a feasible solution to \cref{eq:multi_lp_bin}, and $\obj[\X]{\tilde{H}}\le \obj[\X]{H}$.
\end{proof}